\documentclass[11pt]{article}

\usepackage{tikz}
\usepackage{tikz-qtree}
\usetikzlibrary{arrows,fit,shapes,backgrounds,calc,graphs,graphs.standard,cd}

\tikzset{
    >=latex,                    % use LaTex-style arrowheads
    graph/.style={circle,fill,scale=0.5}, % small dots for most graphs
    smallGraph/.style={circle,fill,scale=0.25}, % smaller dots
    labeled/.style={circle,draw}  % empty circles for labeled graphs
}

\usepackage[utf8]{inputenc}
\usepackage{newunicodechar}
\usepackage[T1]{fontenc}
\usepackage{lmodern}

\usepackage{amsmath,amssymb,amsfonts,mathtools,amsthm}

\usepackage[linesnumbered,vlined]{algorithm2e}
\DontPrintSemicolon

\newcommand{\newData}[2]{\newcommand{#1}{\DataSty{#2}\xspace}}
\newcommand{\newFunc}[2]{\newcommand{#1}{\FuncSty{#2}\xspace}}

\newcommand{\DS}[1]{\DataSty{#1}}

\SetKw{True}{true}
\SetKw{False}{false}

\newcommand{\concept}[1]{\textbf{#1}}

\newcommand{\etal}{\textit{et al.}}

\newcommand{\Tuple}[1]{\left\langle{#1}\right\rangle}

\newunicodechar{¬}{\neg}   % ¬
\newunicodechar{∧}{\wedge} % ∧
\newunicodechar{⊕}{\oplus} % ⊕
\newunicodechar{∨}{\vee}   % ∨
\newunicodechar{→}{\rightarrow}  % →
\newunicodechar{⇒}{\Rightarrow}  % ⇒
\newunicodechar{←}{\leftarrow}   % ←
\newunicodechar{↔}{\leftrightarrow}   % ↔
\newunicodechar{⊢}{\vdash}   % ⊢

\newunicodechar{↦}{\mapsto}  % ↦
\newunicodechar{↪}{\hookrightarrow}  % ↪
\newunicodechar{↣}{\rightarrowtail}  % ↣

\newunicodechar{∀}{\forall} % ∀
\newunicodechar{∃}{\exists} % ∃

\newunicodechar{∩}{\cap}    % ∩
\newunicodechar{∪}{\cup}    % ∪
\newunicodechar{∖}{\setminus}    % ∖
\newunicodechar{×}{\times}    % ×
\newunicodechar{∈}{\in}     % ∈
\newunicodechar{∉}{\not\in}     % ∉
\newunicodechar{∋}{\ni}     % ∋
\newunicodechar{⊆}{\subseteq}    % ⊆
\newunicodechar{⊇}{\supseteq}    % ⊇
\newunicodechar{⊑}{\sqsubseteq}    % ⊑
\newunicodechar{⊒}{\sqsupseteq}    % ⊑
\newunicodechar{≾}{\lesssim}
\newunicodechar{≿}{\grtsim}

\newunicodechar{∘}{\circ}    % ∘
\newunicodechar{⋅}{\cdot}    % ⋅

\newunicodechar{∅}{\emptyset}  % ∅
\DeclareUnicodeCharacter{2115}{\mathbb{N}} % ℕ
\DeclareUnicodeCharacter{2124}{\mathbb{Z}} % ℤ
\DeclareUnicodeCharacter{211A}{\mathbb{Q}} % ℚ
\DeclareUnicodeCharacter{211D}{\mathbb{R}} % ℝ
\DeclareUnicodeCharacter{2102}{\mathbb{C}} % ℂ
\newunicodechar{ℵ}{\aleph}     % ℵ
\newunicodechar{ℶ}{\beth}      % ℶ

\DeclareUnicodeCharacter{2131}{\mathcal{F}} % ℱ

\newunicodechar{∑}{\sum}  % ∑
\newunicodechar{∏}{\prod} % ∏
\newunicodechar{∫}{\int}  % ∫
\newunicodechar{±}{\pm}   % ±

\newunicodechar{≠}{\neq}    % ≠
\newunicodechar{≡}{\equiv}  % ≡
\newunicodechar{≈}{\approx} % ≈
\newunicodechar{≃}{\simeq}  % ≃
\newunicodechar{≅}{\cong}   % ≅
\newunicodechar{≤}{\le}     % ≤
\newunicodechar{≥}{\ge}     % ≥
\newunicodechar{∣}{\divides}% ∣

\newunicodechar{⊤}{\top}    % ⊤
\newunicodechar{⊥}{\bot}    % ⊥
\newunicodechar{∞}{\infty}  % ∞
\newunicodechar{§}{\S}      % §
\newunicodechar{◇}{\Diamond}      % ◇

\newunicodechar{Γ}{\Gamma}  % Γ
\newunicodechar{Δ}{\Delta}  % Δ
\newunicodechar{Θ}{\Theta}  % Θ
\newunicodechar{Λ}{\Lambda} % Λ
\newunicodechar{Ξ}{\Xi}     % Ξ
\newunicodechar{Π}{\Pi}     % Π
\newunicodechar{Σ}{\Sigma}  % Σ
\newunicodechar{Φ}{\Phi}    % Φ
\newunicodechar{Ω}{\Omega}  % Ω
\newunicodechar{α}{\alpha}  % α
\newunicodechar{β}{\beta}   % β
\newunicodechar{γ}{\gamma}  % γ
\newunicodechar{δ}{\delta}  % δ
\newunicodechar{ε}{\epsilon} % ε
\newunicodechar{ζ}{\zeta}   % ζ
\newunicodechar{η}{\eta}    % η
\newunicodechar{κ}{\kappa}  % κ
\newunicodechar{λ}{\lambda} % λ
\newunicodechar{μ}{\mu}     % μ
\newunicodechar{ξ}{\xi}     % ξ
\newunicodechar{π}{\pi}     % π
\newunicodechar{ρ}{\rho}    % ρ
\newunicodechar{σ}{\sigma}  % σ
\newunicodechar{τ}{\tau}    % τ
\newunicodechar{φ}{\phi}    % φ
\newunicodechar{χ}{\chi}    % χ
\newunicodechar{ψ}{\psi}    % ψ
\newunicodechar{ω}{\omega}  % ω

\DeclarePairedDelimiter{\card}{\lvert}{\rvert}

\DeclarePairedDelimiter{\floor}{\lfloor}{\rfloor}

\newcommand{\ComplexityClass}[1]{\ensuremath{\mathrm{#1}}}
\newcommand{\DeclareComplexityClass}[2]{\newcommand{#1}{\ComplexityClass{#2}}}

\DeclareComplexityClass{\classNP}{NP}
\DeclareComplexityClass{\classcoNP}{coNP}
\DeclareComplexityClass{\classEXP}{EXP}
\DeclareComplexityClass{\classNEXP}{NEXP}
\DeclareComplexityClass{\classP}{P}
\DeclareComplexityClass{\classPpoly}{P/poly}
\DeclareComplexityClass{\classR}{R}
\DeclareComplexityClass{\classRP}{RP}
\DeclareComplexityClass{\classcoR}{coR}
\DeclareComplexityClass{\classcoRP}{coRP}
\DeclareComplexityClass{\classZPP}{ZPP}
\DeclareComplexityClass{\classBPP}{BPP}
\DeclareComplexityClass{\classBPL}{BPL}
\DeclareComplexityClass{\classBQP}{BQP}
\DeclareComplexityClass{\classPP}{PP}
\DeclareComplexityClass{\classSharpP}{\#P}
\DeclareComplexityClass{\classParityP}{⊕P}
\DeclareComplexityClass{\classL}{L}
\DeclareComplexityClass{\classNL}{NL}
\DeclareComplexityClass{\classcoNL}{coNL}
\DeclareComplexityClass{\classAL}{AL}
\DeclareComplexityClass{\classTIME}{TIME}
\DeclareComplexityClass{\classSPACE}{SPACE}
\DeclareComplexityClass{\classNTIME}{NTIME}
\DeclareComplexityClass{\classcoNTIME}{NTIME}
\DeclareComplexityClass{\classNSPACE}{NSPACE}
\DeclareComplexityClass{\classASPACE}{ASPACE}
\DeclareComplexityClass{\classcoNSPACE}{coNSPACE}
\DeclareComplexityClass{\classATIME}{ATIME}
\DeclareComplexityClass{\classAP}{AP}
\DeclareComplexityClass{\classPSPACE}{PSPACE}
\DeclareComplexityClass{\classNPSPACE}{NPSPACE}
\DeclareComplexityClass{\classAPSPACE}{APSPACE}
\DeclareComplexityClass{\classIP}{IP}
\DeclareComplexityClass{\classDelta}{\Delta}
\DeclareComplexityClass{\classSigma}{\Sigma}
\DeclareComplexityClass{\classPi}{\Pi}
\DeclareComplexityClass{\classPH}{PH}
\DeclareComplexityClass{\classRL}{RL}
\DeclareComplexityClass{\classSL}{SL}
\DeclareComplexityClass{\classAC}{AC}
\DeclareComplexityClass{\classNC}{NC}
\DeclareComplexityClass{\classFP}{FP}
\DeclareComplexityClass{\classUP}{UP}
\DeclareComplexityClass{\classESO}{ESO}
\DeclareComplexityClass{\classSO}{SO}
\DeclareComplexityClass{\classPCP}{PCP}
\DeclareComplexityClass{\classFO}{FO}
\DeclareComplexityClass{\classACC}{ACC}
\DeclareComplexityClass{\classTC}{TC}
\DeclareComplexityClass{\classPromiseBPP}{promiseBPP}
\DeclareComplexityClass{\classPromiseP}{promiseP}

\DeclareComplexityClass{\BPTISP}{BPTISP}
\DeclareComplexityClass{\BPL}{BPL}
\DeclareComplexityClass{\BP}{BP}

\DeclareComplexityClass{\POP}{POP}
\DeclareComplexityClass{\CRN}{CRN}
\DeclareComplexityClass{\GOSSIP}{GOSSIP}
\DeclareComplexityClass{\PULL}{PULL}
\DeclareComplexityClass{\PUSH}{PUSH}
\DeclareComplexityClass{\MATCHING}{MATCHING}
\DeclareComplexityClass{\SHUFFLE}{SHUFFLE}
\DeclareComplexityClass{\BETWEEN}{BETWEEN}

\newtheorem{theorem}{Theorem}[section]
\newtheorem{lemma}[theorem]{Lemma}

\newtheorem{definition}[theorem]{Definition}

\numberwithin{equation}{section}

\DeclareMathOperator{\E}{E}

\DeclareMathOperator{\Pred}{Pred}

\newcommand{\Prob}[1]{\Pr\left[{#1}\right]}
\newcommand{\ProbCond}[2]{\Pr\left[{#1}\;\middle\vert\;{#2}\right]}
\newcommand{\Exp}[1]{\E\left[{#1}\right]}
\newcommand{\Set}[1]{\left\{{#1}\right\}}
\newcommand{\SetWhere}[2]{\left\{{#1}\;\middle\vert\;{#2}\right\}}

\newcommand{\divides}{\mid}

\DeclarePairedDelimiter{\parens}{(}{)}

\newData{\Role}{role}
\newData{\Follower}{follower}
\newData{\Clock}{clock}
\newData{\Init}{init}
\newData{\FindLeft}{find-left}
\newData{\FindRight}{find-right}
\newData{\Winnow}{winnow}
\newData{\Move}{move}
\newData{\Done}{done}
\newData{\Mark}{mark}
\newData{\Left}{left}
\newData{\Right}{right}
\newData{\State}{state}
\newData{\TapeSymbol}{symbol}
\newData{\Count}{count}
\newData{\Direction}{direction}
\newData{\Halt}{halt}
\newData{\TM}{TM}
\newData{\Coin}{coin}
\newData{\Heads}{heads}
\newData{\Tails}{tails}
\newFunc{\Step}{step}

\SetKwFor{Procedure}{procedure}{}{end procedure}

\title{Stochastic well-structured transition systems}

\author{James Aspnes\\Yale University}

\date{\quad}

\begin{document}

\maketitle

\begin{abstract}

    Extending well-structured transition
    systems~\cite{Finkel1987,FinkelS2001} to incorporate a
    probabilistic scheduling rule, we define a new class of
    \emph{stochastic} well-structured transition systems that includes
    population protocols~\cite{AngluinADFP2006}, chemical reaction
    networks~\cite{Gillespie1977}, and many common gossip models; as
    well as augmentations of these systems by an oracle that exposes a
    total order on agents as in population protocols in the comparison
    model~\cite{GanczorzGJS2023} or an equivalence relation as in
    population protocols with unordered data~\cite{BlondinL2023}.

    We show that any implementation of a phase clock
    in these systems either stops or ticks too fast after polynomially
    many expected steps, and that any terminating computation in these
    systems finishes or fails in expected polynomial time.
    This latter property allows 
    an exact characterization of the computational
    power of many stochastic well-structured
    transition systems augmented with a total order or
    equivalence relation on agents, showing that these compute exactly the
    languages in $\classBPP$, while the corresponding unaugmented
    systems compute just the symmetric languages in $\classBPL$.
\end{abstract}

\section{Introduction}

We define a class of \concept{stochastic well-structured transition
systems} (SWSTSs) that includes many common models of distributed systems
involving finite-state agents with random scheduling, and show that
systems in this class that attempt to delay generating a signal for
some amount of time can only do so if for a time polynomial in the
number of agents. The key idea is to extend the notion of a
\concept{well-structured transition
system} (WSTS)~\cite{Finkel1987,FinkelS2001}, a system whose configurations
from a well-quasi-order that is a simulation with respect to the
transition relation, by adding probabilities on transitions that are
polynomial in a measure of the size of the configuration.

This extension allows proving that stochastic well-structured
transition systems are bad at waiting: given an upward-closed target
set, any configuration can either reach the target set in a constant
number of steps with polynomial probability or cannot reach it at all.
This has implications for building clocks (any clock protocol
eventually fails or advances too fast after an expected polynomial number of
ticks) and terminating protocols (any protocol that explicitly
terminates eventually fails or terminates after expected polynomial
time). We use these results to give an exact characterization of the
computational power of a variety of models that fit within our
framework.

Though we are concerned with random scheduling,
the present work is inspired by previous work analyzing population
protocols~\cite{AngluinADFP2006} with \emph{adversarial} scheduling using
well-quasi-orders. It has previously been observed in this context
that population protocols are instances of well-structured transition
systems~\cite{EsparzaGLM2016}, and the use of Dickson's
Lemma~\cite{Dickson1913} to show that the configuration spaces of
population protocols can be well-quasi-ordered is a key step in
the classic demonstration by Angluin~\etal~\cite{AngluinAER2007} that
population protocols stably computed precisely the semilinear
predicates, the lower bound of Doty and Soloveichik on leader election
in population protocols~\cite{DotyS2018}, the proof by Chen~\etal~\cite{ChenCDS2017} that chemical
reaction networks for many problems are subject to inherent
bottlenecks called \concept{speed faults}, and the result of Mathur and
Ostrovsky~\cite{MathurO2022} characterizing the power of
self-stabilizing population protocols. To the best of our knowledge,
our extension of well-structured transition systems to a model
representing randomized scheduling is new.

The paper is organized as follows. In Section~\ref{section-WSTS},
we discuss well-structured transition systems, including a discussion of
well-quasi-orders and the connections between well-structured
transition systems, population protocols, and demonstrating
reachability. In Section~\ref{section-stochastic-WSTS}, we define
stochastic well-structured transition systems. Our main results about
waiting are shown in Section~\ref{section-finish-or-fail}, including
the central property that stochastic well-structured transition systems
can only wait for a polynomial number of steps
(Theorem~\ref{theorem-finish-or-fail}), as well as several
consequences for clocks and terminating computations.
A catalog of distributed computing models with randomized scheduling
that yield SWSTSs is given in Section~\ref{section-examples};
these include \concept{population protocols}~\cite{AngluinADFP2006},
closed instances of \concept{chemical reaction networks}~\cite{Gillespie1977},
various synchronous \concept{gossip}
models~\cite{BertierBK2009,BechettiCNPS2014,BoczkowskiNFK2018,ArchivioKNV2024},
the synchronous token-passing model of~\cite{AlistarhGR2022},
the \concept{comparison model} of~\cite{GanczorzGJS2023}, \concept{unordered population
protocols}~\cite{BlondinL2023}, and others.
Section~\ref{section-computation} provides some additional definitions
needed to describe computation by SWSTSs.
Section~\ref{section-simulating-SWSTS} gives uniformity conditions on
SWSTSs that allow such computations to be simulated in either
bounded-error polynomial time ($\classBPP$) or bounded-error
polynomial-time and logarithmic space ($\classBPL$). Both cases
rely on Theorem~\ref{theorem-finish-or-fail} to show that the
simulated protocol necessarily finishes in polynomially many expected
steps. Section~\ref{section-simulating-TM} proves converse results for
many of the models described in Section~\ref{section-examples},
showing that these models compute precisely the
predicates in $\classBPP$ or the symmetric predicates in $\classBPL$,
depending on how much structure each model imposes on its components.
Possibilities for future work are discussed in
Section~\ref{section-conclusion}.

\section{Well-structured transition systems}
\label{section-WSTS}

The configurations of distributed systems can often be quasi-ordered
by some sort of embedding relation, where $s≤t$ if $s$ can be viewed
as a subconfiguration of $t$, and the relation $≤$ being a quasi-order simply means
that it is reflexive and transitive.\footnote{This differs from a
partial order in that having \concept{equivalent} elements $x$ and $y$
with $x≤y$ and $y≤x$ does not necessarily imply
$x=y$. We will see examples of this when we consider embeddings that
may permute agents without adding or removing any.} 
We are particularly interested in models where the embedding relation is
a \concept{well-quasi-order}~\cite{Kruskal1972} or \concept{wqo}, a quasi-order that has
no infinite descending chains or infinite antichains.\footnote{See
§\ref{section-wqos} for more details about wqos.}

If the embedding is a wqo and
the transitions in a system appropriately respect the embedding,
we get a \concept{well-structured transition
system}~\cite{Finkel1987,FinkelS2001} or \concept{WSTS}.
Formally, a WSTS $(S,≤,→)$ consists of a set of states $S$ and two binary relations $≤$ and $→$, where
$≤$ is a well-quasi-order on $S$ and $→$ is a
\concept{transition relation} that is \concept{compatible} with $≤$ in
the sense that whenever $s→s'$ and $s≤t$, there exists $t'$ such
that $t→t'$ and $s'≤t'$.\footnote{This definition, due to Finkel and
Schnoebelen~\cite{FinkelS2001}, is a streamlined version of the
original definition of Finkel~\cite{Finkel1987}. But as it is now
standard, we will use it here. See also Schmitz and
Schnoebelen~\cite{SchmitzS2013}, who provide some additional
motivation for the definition and discuss subsequent developments of
the theory.}

When $s≤t$, $s→s'$, $s'≤t$, $t→t'$, we will say that the transition
$s→s'$ \concept{occurs} in $t→t'$.

We will mostly be considering transition systems where we have a
finite collection of
agents each with a state from some finite set $Q$. This makes a state
of the system as a whole, which we will refer to as a
\concept{configuration} to distinguish it from the states of individual agent,
an element of $Q^*$, possibly with some additional structure. The
main idea is to let $s≤t$ if there is an injection $f$ mapping agents
in $s$ to agents in $t$ that preserves any necessary structure in the
configuration. In this case, the compatibility property says that if
there exists an injection $f:s→t$ and a transition $s→s'$, then there
is a configuration $t'$ with an injection $f':s'→t'$ and a transition
$t→t'$. For many models, the easy way to find a transition $t→t'$ is
just to let each agent in the image set $f(s)$ do whatever its
corresponding agent did in $s$, and let any remaining agents in $t$
either do nothing (if permitted by the model) or do something that
doesn't affect agents in $f(s)$. In this case, the fact that $s→s'$
occurs in $t→t'$ just means that the agents in $f(s)$ observe a
transition that looks locally like $s→s'$.

\subsection{Example: Population protocols}
\label{section-pop-WSTS}

A simple example of a well-structured transition system
is a (finite-state) \concept{population
protocol}~\cite{AngluinADFP2006}.
In the population protocol model, we have $n$ agents with states in
$Q$, giving configurations in $Q^n$.
A \concept{step} $s→s'$ of a population
protocol consists of choosing two agents $i$ and $j$ in $s$, and
updating their states according to a \concept{joint transition
function} $δ:Q^2→Q^2$, so that $(s'_i,s'_j) = δ(s_i,s_j)$ and $s'_k =
s_k$ for all $k∉\Set{i,j}$. A simple example of a population protocol
is the cancellation-based
leader election protocol~\cite{AngluinADFP2006} with two states $L$ (leader)
and $F$ (follower), with a single non-trivial transition $L,L → L,F$
that prunes excess leaders. Note that interactions in a population
protocol are not symmetric: the first agent in the pair (the
\concept{initiator}) is distinguished from the second (the
\concept{responder}) and the transition relation can take this
distinction into account.

An embedding $s≤t$ for population protocols is an order-preserving injective map $f$ from
the indices of $s$ to the indices of $t$, such that $s_i = t_{f(i)}$
for each $i$. That this relation $≤$ is a quasi-order follows from
Higman's Lemma~\cite{Higman1952}.
To show that $≤$ and $→$ are compatible, for any $s→s'$
and $s≤t$, construct a transition $t→t'$ by taking the agents
$i$ and $j$ that interact in $s$ and applying the same interaction to
agents $f(i)$ and $f(j)$ in $t$.

\subsection{Identifying well-quasi-orders}
\label{section-wqos}

Because well-structured transition systems are so closely connectecd
to well-quasi-orders, it will be useful to establish some further
characterizations of well-quasi-orders.

There are several equivalent definitions of a
well-quasi-order~\cite{Kruskal1972}. 
The most widely used is that a
$≤$ is a wqo if it is a quasi-order (a relation that is reflexive and
transitive) such that in any infinite sequence of elements $x_0,
x_1, x_2, \dots$ there exist indices $i < j$ such that $x_i ≤ x_j$.
This excludes both infinite antichains (where all values are
incomparable with each other) and infinite descending chains (where
each value is strictly smaller than its predecessor). The latter
condition provides the connection to the more common notion of well-orderings.

An equivalent definition that can be more difficult to check but 
easier to use is that $≤$ is a
well-quasi-order if it is a quasi-order with the \concept{finite basis
property}~\cite{Higman1952}, which says that any
\concept{upward-closed} set is the \concept{upward closure} of a
finite set.\footnote{Here a set $T$ is upward-closed if $x∈T$ and $x≤y$ implies
$y∈T$, and the upward-closure of a set is its smallest upward-closed
superset.} The finite basis property is the key idea behind the study of well-structured
transition systems. It means that 
reachability of upward-closed sets of states, even in an infinite state space, can be
reduced to reachability of a finite basis of those sets.

Well-quasi-orders have a long history, and several classical results
exist showing that particular embedding relations yield
well-quasi-orders. We will mostly be using the following well-known
facts:

\begin{itemize}
    \item Let $Q$ be a finite set and let $≤$ be a quasi-order on
        $Q$.
        Then $Q$ is a well-quasi-order. (Proof: Any infinite sequence
        of elements of $Q$ contains duplicates $x_i = x_j$ with $i<j$.)
        
        A common special case is when $≤$ is the minimal quasi-order
        where $a≤b$ if and only if $a=b$.
    \item Let $(Q,≤)$ be a well-order. The $(Q,≤)$ is a
        well-quasi-order. (Proof: Any nonempty set in a well-order has a unique minimal element,
        giving a finite basis of size one for any upward-closed set.)
    \item \concept{Dickson's Lemma}~\cite{Dickson1913}:
        Let $Q$ be a finite set. Then component-wise ordering on vectors in
        $ℕ^Q$, where $x≤y$ if $x_i ≤ y_i$ for all $i$, is a
        well-quasi-order.
    \item 
        \concept{Higman's Lemma}~\cite{Higman1952}:
        Let $(Q,≤)$ be a well-quasi-order. Given finite sequences $x,y∈Q^*$, let an 
        \concept{embedding} of $x$ into $y$
        be a strictly increasing function $f$ from indices of $x$ to
        indices of $y$ such that $x_i ≤ y_{f(i)}$ for all $i$, and define
        $x≤y$ if there exists an embedding of $x$ into $y$. Then $≤$ is
        a well-quasi-order.
\end{itemize}

\subsection{Reachability}

Following~\cite{FinkelS2001}, we will use the notation
$s \stackrel{k}{→} t$ to say that there exists a sequence $s = s_0 →
s_1 → \dots s_k = t$ and $s \stackrel{*}{→} t$ to say that $s
\stackrel{k}{→} t$ for some $k$. We will also write
$s\stackrel{≤k}{→}t$
when $s \stackrel{\ell}{→} t$ for some $0 ≤ \ell ≤ k$.
Equivalently, we will say that $t$ is reachable in $k$ steps, is
reachable, or is reachable in at most $k$ steps in each of these
cases. When $T$ is a set, we will write $s \stackrel{k}{→} T$ 
if there exists some $t∈T$ such that $s \stackrel{k}{→} t$, and
similarly for the other reachability conditions.

We will use the notation
\begin{align*}
    \Pred(T) &= \SetWhere{s}{s→T}
    \\
    \Pred^k(T) &= \SetWhere{s}{s\stackrel{k}{→}T},
    \\
    \Pred^{≤k}(T) &= \SetWhere{s}{s\stackrel{≤k}{→}T},
    \intertext{and}
    \Pred^*(T) &= \SetWhere{s}{s\stackrel{*}{→}T}.
\end{align*}

Much of the usefulness of WSTSs comes from the following standard lemma:
\begin{lemma}
    \label{lemma-predecessors}
    Let $(S,≤,→)$ be a well-structured transition system. Then for any
    upward-closed $T⊆S$ and any $k$, $\Pred(T)$, $\Pred^k(T)$, 
    $\Pred^{≤k}(t)$, and
    $\Pred^*(T)$ are all upward-closed.
\end{lemma}
\begin{proof}
    For $\Pred(T)$, observe that $s∈\Pred(T)$ means that there is some
    $t∈T$ such that $s→t$. But then if $s'≥s$, there is some $t'≥t$
    with $s'→t'$. Having $t'≥t$ puts $t'$ in $T$ (which is
    upward-closed); this implies that $s'$ is in $\Pred(T)$.

    This result extends to $\Pred^k(T)$ by a straightforward induction and to
    $\Pred^{≤k}(T)$ and $\Pred^*(T)$ by observing that the property of
    being upward-closed is preserved by unions.
\end{proof}

When $T$ is upward-closed, the sets $\Pred^{≤0}(T) ⊆ \Pred^{≤1}(T) ⊆
\dots$ form an increasing sequence of upward-closed sets. It is
straightforward to show
(see, for example, \cite[Lemma 3.4]{AbdullaCJT1996}) that
for any wqo, any such sequence eventually stabilizes with
$\Pred^{≤k+1}(T) = \Pred^{≤k}(T)$ for some $k$, 
meaning that for any $T$,
there is some finite $k$ such that 
$\Pred^*(T) = \Pred^{≤k}(T)$.

%\footnote{The proof of this is that since $\Pred^*(T)$ has a finite
%basis $\Set{x_1,\dots,x_m}$, any $s$ in $\Pred^*(T)$ lies above some
%basis element $x_i$. But then for some $k_i$ there is an element $t$
%of $T$ and a sequence $x_i \stackrel{k_i}{→} t$ that extends via 
%compatibility to a sequence $s \stackrel{k_i}{→} t'$ for some $t'≥t$,
%putting $s$ in $\Pred^{k_i}(T)$. Now let $k = \max k_i$ to put all
%such $s$ in $\Pred^{≤k}(T)$.}

This is a very strong property, which forms the basis for much of the
usefulness of WSTSs in verification (see~\cite{AbdullaCJT1996} for
examples). We can restate it as:
\begin{lemma}
    \label{lemma-maximum-distance}
    Let $(S,≤,→)$ be a well-structured transition system, and let $T⊆S$ be
    upward-closed. Then there is a fixed $k$ such that 
    for any $s∈S$, either $s$ can reach $T$ in at most $k$ steps, or $s$
    cannot reach $T$ at all.
\end{lemma}

\section{Stochastic well-structured transition systems}
\label{section-stochastic-WSTS}

The standard notion of well-structured transition systems does not incorporate an
explicit scheduling mechanism: the $→$ relation defines which
configurations \emph{can} follow each other but says nothing about how
likely it is that these
configurations \emph{will} follow each other.
This is acceptable when one is interested primarily in reachability,
but we would like to extend the standard definition to cover models
with explicit transition probabilities.
Many of the models we are considering have a probabilistic
scheduling rule that makes the choice of transition at each step
random according to some distribution.

\begin{definition}
    \label{definition-SWSTS}
    A \concept{stochastic well-structured
    transition system} (\concept{SWSTS}) 
    $(S,≤,→,\card{⋅},\Pr)$
    consists of:
    \begin{itemize}
        \item A WSTS $(S,≤,→)$.
        \item A function $\card{⋅}$ that assigns to each configuration
            $s∈S$ a \concept{weight} $\card{s}$.
        \item A function $\Pr$ that assigns a \concept{transition
            probability}
            $\Prob{s→s'}$ to each transition $s→s'$,
             making an execution
            $X_0→X_1→X_2→\dots$ of the system into a trajectory of
            a Markov chain with $\ProbCond{X_{i+1} = s'}{X_i = s} =
            \Prob{s→s'}$.
    \end{itemize}
    satisfying the requirement that:
    \begin{itemize}
        \item 
            For each transition $s→s'$, there is a constant $k$, such that
            in any execution of the system,
            $\ProbCond{X_{i+1}≥s'}{X_i≥s} = Ω\parens*{\card{X_i}^{-k}}$.
    \end{itemize}
\end{definition}

We refer to the last condition as having 
\concept{polynomial transition probabilities}. Note that there does not need to be a
single transition $t→t'$ extending $s→s'$; it is enough that starting
from any $t≥s$ that the set of possible transitions $t→t'$ with $t'≥s'$ has
polynomial probability in total.

An example of an SWSTS is a population protocol. 
With randomized scheduling, the weight
$\card{X_i}$ of a configuration $X_i$ is the number of agents $n$ in
$X_i$. Since each possible pairwise interaction occurs with the same
probability $\frac{1}{n(n-1)}$, whatever pair of agents interact in a
transition $s→s'$ will also interact in any $X_i$ with $s≤X_i$ with the same
probability $\frac{1}{n(n-1)} = Ω\parens*{\card{X_i}^{-2}}$
In this case the exponent $k$ is fixed, but in general the
choice of $k$ may depend on the choice of transition $s→s'$.

Often we will ask that the weight function is preserved by
transitions: if $s→s'$, then $\card{s} = \card{s'}$. An SWSTS with a
weight function with this property will be called \concept{closed}.
Examples of closed SWSTSs include most distributed systems where a
transition only changes the states of agents and not their identities
or number, and chemical reaction networks~\cite{Gillespie1977} where
the weight is defined as the total atomic number and any
transition preserves this quantity. Examples of non-closed
SWSTSs might be systems that allow new agents to arrive over time or
chemical reaction networks that increase their volume by
absorbing feedstocks that are not otherwise counted in the weight.

\section{Reachability in polynomial expected steps}
\label{section-finish-or-fail}

Random scheduling allows us to compute the expected time until a
protocol reaches a desired configuration. In this section, we show
that executions of closed SWSTSs
finish or fail in a polynomial
number of steps on average, when finishing is characterized by
reaching some upward-closed set $T$ and failing is characterized by
reaching a configuration that will never reach $T$.

We start with a more limited result, which is that any reachable upward-closed set
$T$ is reached in a constant number of steps with polynomial
probability:
\begin{lemma}
    \label{lemma-polynomial-path}
    Let $X$ be an execution of a closed SWSTS
    $(S,≤,→,\card{⋅},\Pr)$.
    Let $T⊆S$ be upward-closed. Let $X_i ∈ \Pred^*(T)$. Then there are fixed constants
    $\ell$ and $k$ such that $\Prob{∃j≤\ell: X_{i+j} ∈ T} ≥
    Ω\parens*{\card{X_t}^{-k}}$.
\end{lemma}
\begin{proof}
    By Lemma~\ref{lemma-predecessors}, $\Pred^*(T)$ is upward-closed,
    and by Lemma~\ref{lemma-maximum-distance}, $\Pred^*(T) =
    \Pred^{≤\ell}(T)$ for some fixed $\ell$. 
    Choose some finite basis for $\Pred^{≤\ell}(T)$.

    Any $t$ in $\Pred^*(T)=\Pred^{≤\ell}(T)$ lies above some
    basis element $s$ of $\Pred^{≤\ell}(T)$. So there is an $\ell_s
    ≤ \ell$ and sequence of transitions
    $s = s_0 → s_1 → \dots → s_{\ell_s} ∈ T$ that 
    puts $s$ in $\Pred^{≤\ell}(T)$.
    This maps to a family of partial executions
    $X_i → X_{i+1} → \dots → X_{i+\ell_s} ∈ T$ where each $X_{i+j} ≥
    s_j$. Applying polynomial transition probabilities, there are
    fixed constants $k_j$ such that
    \begin{align*}
        \ProbCond{X_{i+\ell_s}≥s_{\ell_s}}{X_i≥s}
        &≥ ∏_{j=1}^{\ell_s} \ProbCond{X_{i+j}≥s_j}{X_{i+j-1}≥s_{j-1}}
        \\
        &≥ ∏_{j=1}^{\ell_s} Ω\parens*{\card{X_i}^{-k_j}}
        \\
        &= Ω\parens*{\card{X_i}^{-k_s}},
        \intertext{where}
        k_s &= ∑_{j=1}^{\ell_s} k_j.
    \end{align*}
    Let $k=\max_s k_s$, where $s$ ranges over the finite basis of
    $\Pred^{≤\ell}(T)$. Then
    \begin{align*}
        \ProbCond{∃j ≤ \ell: X_{i+j} ∈ T}{X_i ∈ \Pred^*(T)}
        &≥ \min_s \ProbCond{X_{i+\ell_s} ∈ T}{X_i≥s}
        \\
        &≥ Ω\parens*{\card{X_i}^{-k}}.
    \end{align*}
\end{proof}

An immediate consequence of the lemma is the following,
which is our central result for closed SWSTSs:
\begin{theorem}
    \label{theorem-finish-or-fail}
    Let $X$ be an execution of a closed SWSTS $(S,≤,→,\card{⋅},\Pr)$.
    Let $T⊆S$ be upward-closed.
    Let $τ$ be the first time at which $X_τ ∈ T$ or $X_τ ∉
    \Pred^*(T)$.
    Then there is a constant $k$ such that $\Exp{τ} =
    O\parens*{\card{X_0}^k}$.
\end{theorem}
\begin{proof}
    From Theorem~\ref{theorem-finish-or-fail}, there are constants $\ell$
    and $k$ such that if $X_i ∈ \Pred^*(T)$, there is an
    $Ω\parens*{\card{X_i}^{-k}} = Ω\parens*{\card{X_0}^{-k}}$ chance that $X_{i+j} ∈ T$ for some $j
    ≤ \ell$. So every $\ell$ steps, $X_i$ either enters $T$ with
    probability $Ω\parens*{\card{X_0}^{-k}}$ or leaves $\Pred^*(T)$
    forever. In either case the expected waiting time for this event is
    bounded by
    $O\parens*{\ell \card{X_0}^k}
    = O\parens*{\card{X_0}^k}$.
\end{proof}

\subsection{Failure of phase clocks}

In the context of population protocols, a 
\concept{phase clock}~\cite{AngluinAE2008fast}
is a mechanism that allows a protocol
to wait for some number of steps with high probability, typically to
allow some other stochastic process (like an epidemic) to finish.
Phase clocks in the
literature~\cite{AngluinAE2008fast,KosowskiU2018,AlistarhAG2018,GasieniecS2018,GasieniecS2020,GasieniecSS2021}
repeatedly count off intervals of expected $Θ(n \log n)$ or some
other number of steps, where the start
of each interval is marked by a configuration containing a leader
agent (or sometimes multiple leader agents) entering some class of
states that we will call \concept{tick}, after which the leader reverts to
some other class of non-tick states we will call \concept{tock}. For a
phase clock with this structure, Lemma~\ref{lemma-polynomial-path}
guarantees that it will fail by ticking too fast within a polynomial
number of ticks.

The argument is simple: The class of configurations $T$ containing at
least one tick state and that class of configurations $T'$ containing
at least one tock state are both upward closed. So
Lemma~\ref{lemma-polynomial-path} says that from any configuration in
$T'$ there is a polynomial probability of reaching a configuration in
$T$ in a constant number of steps, and a second application of the
lemma gives a similarly polynomial probability of returning to a
configuration in $T$ in a constant number of steps. The only way to
avoid this event over repeated ticks of the clock is through failure:
reaching a configuration from which $T$ and $T'$ are no longer
reachable.

We can state this outcome formally as:
\begin{theorem}
    \label{theorem-phase-clock}
    Define a phase clock protocol in an SWSTS as a protocol repeatedly
    alternates between configurations in an upward-closed set $T$ and
    an upward-closed set $T'$, and define a round of the phase clock
    as a sequence of transitions $T\stackrel{*}{→}T'\stackrel{*}{→}T$.
    Then for a phase clock in any closed SWSTS, the expected time until 
    it performs a round in $O(1)$ steps or reaches a configuration
    that can no longer reach $T$ is polynomial in the size of
    its initial configuration.
\end{theorem}
\begin{proof}
    Starting in $X_0$, run the phase clock until it reaches $T$ or a
    configuration that can't reach $T$; this
    takes time polynomial in $\card{X_0}$ by
    Theorem~\ref{theorem-finish-or-fail}.

    From any configuration in $T$, applying
    Lemma~\ref{lemma-polynomial-path} gives a polynomial probability
    of reaching $T'$ and then $T$ again in a constant number of steps.
    If this does not occur, apply Theorem~\ref{theorem-finish-or-fail}
    twice to either carry out a round of the phase clock or reach a
    state that can no longer reach $T$, in polynomial time in either
    case.

    If we reach $T$ again, this gives another polynomial chance of a
    constant-time round. So after an expected polynomial number of
    iterations, each of which takes polynomial expected time, we
    either get a constant-time round or reach a configuration that
    can't reach $T$.
\end{proof}

\subsection{Polynomial-time termination of SWSTS computations}

A further example of where Theorem~\ref{theorem-finish-or-fail} is useful is
limiting the expected running time of closed SWSTS protocols where
termination is signaled by a special leader agent entering one of
finitely many terminal states $L_{\textrm{decision}}$. Since any configuration that
contains such a state is terminal, the class $T$ of terminal
configurations is upward-closed. This implies that either a terminal
configuration is reached or all terminal configurations become
unreachable in polynomial expected time. This puts an inherent limit
on the computational power of, for example, leader-based population
protocols~\cite{AngluinAE2008fast}, showing that they cannot run for
more than a polynomial expected number of steps no matter how cleverly
they are constructed.

We will give numerous examples of SWSTSs in
§\ref{section-examples}. After developing some machinery to describe
computations in these models in §\ref{section-computation}, we will
apply Theorem~\ref{theorem-finish-or-fail} to show that any language
computed in most of these models with bounded two-sided error is in the class
$\classBPP$ of languages computed with bounded two-sided error by
probabilistic polynomial-time Turing machines. This is true even
without imposing a polynomial-time restriction on the SWSTS
computations; instead, Theorem~\ref{theorem-finish-or-fail}
supplies the polynomial-time
restriction for us. However, for this to work, we need to have
computations that terminate explicitly, by entering into a decision
configuration contained in some upward-closed set.

\subsection{Explicit termination vs implicit convergence}

Not all models make the assumption of explicit termination.
\concept{Stable computation} as defined by
Angluin~\etal~\cite{AngluinADFP2006} requires eventual permanent
convergence to configurations in which all agents show the correct
output. This does not require that agents know that they have the
correct output, and indeed we can show that the polynomial-time
guarantee of Theorem~\ref{theorem-finish-or-fail} 
does not apply to convergence, since it is
possible to construct simple population protocols that converge only
after exponential time.

For example, consider an epidemic process with agents that are
\concept{susceptible} ($S$) or \concept{infected} ($I$), supplemented
by a single \concept{doctor} agent ($D$) that cures any infected agent
it meets. Formally, this is a population protocol
with non-trivial transitions $IS → II$ and $ID
→ SD$. Assuming we start with a single doctor, the unique reachable stable
configuration for this protocol consists of the doctor in state $D$
and the other $n-1$ agents all in state $S$.

However, with $k$ out of $n$ agents in state $I$, the probability that
a new agent is infected on the next step is $\frac{k(n-k-1)}{n(n-1)}$,
while the probability that an agent ceases to be infected is only
$\frac{k}{n(n-1)}$. This makes the process a biased random walk with a
bias strongly away from $0$ for most values of $k$. So it will take
a while to reach $0$:
\begin{theorem}
    \label{theorem-doctor-convergence}
    Starting from a configuration with $n-1$ agents in state $I$ and
    one agent in state $D$, the expected time for the $IS→II,ID→SD$ 
    protocol to converge to its
    stable configuration is exponentially small as a function of $n$.
\end{theorem}
\begin{proof}
    Let $X_t$ be the number of infected agents after $t$ steps.
    Conditioning on $X_{t+1} ≠ X_t$, we have
    \begin{align*}
        \ProbCond{X_{t+1} = X_t+1}{X_{t+1}≠X_t}
        &= \frac{n-k-1}{n-k}
        \intertext{and}
        \ProbCond{X_{t+1} = X_t-1}{X_{t+1}≠X_t}
        &= \frac{1}{n-k}.
    \end{align*}

    We will bound the convergence time for $X_t$ by considering a
    simpler coupled process $Y_t≤X_t$ that acts as an intermittent random
    walk with fixed bias, a reflecting barrier at
    $b=\floor{\frac{n-1}{2}}$, an absorbing barrier at $0$, and an
    initial state $Y_0=b≤X_0$.

    The coupling between $Y_t$ and $X_t$ will have $Y_{t+1}≠Y_t$ only
    if $X_{t+1}≠X_t$. When $Y$ does move, we will let
    $\ProbCond{Y_{t+1}=Y_t+1}{Y_{t+1}∉\Set{0,Y_t,b}} = p =
    1-\frac{2}{n-1}$. Because of the reflecting barrier, $Y$ can only
    rise if $Y_t<b$. If $X_t≥b+1$, choosing $Y_{t+1}$ and $X_{t+1}$
    independently preserves $Y_{t+1}≤X_{t+1}$. If $X_t≤b$, then the
    probability that $X_{t+1}$ drops is at most $\frac{1}{n-b} ≤
    \frac{2}{n-1}=1-p$. So in this case we can arrange for $Y_{t+1}$
    to drop whenever $X_{t+1}$ does, again preserving
    $Y_{t+1}≤X_{t+1}$. 

    Now let us consider the expected time for $Y_t$ to reach $0$.
    Whenever $Y_t=b$, it always takes one step to get $Y_{t+1}=b-1$.
    From this point on we have a standard ruin problem where the
    probability that $Y$ next reaches $b$ before reaching $0$ is
    exactly 
    \begin{align*}
        \frac{(q/p)^{b-1} - (q/p)^{b}}{1-(q/p)^b}
    \end{align*}
    (see, for example, \cite[XIV.2.4]{Feller1968}).
    In this case,
    \begin{align*}
        (q/p) &= \frac{2/(n-1)}{1-2/(n-1)} = Θ(1/n),
        \intertext{which gives a probability of reaching $0$ first of}
        \frac{Θ(n^{-b-1}) - Θ(n^{-b})}{1-Θ(n^{-b}}
        &= Θ(n^{-b}) = Θ\parens*{n^{-Θ(n)}}.
    \end{align*}
    Since this event is exponentially improbable, the expected number
    of attempts to reach $0$ from $b-1$ is exponential. Each such
    attempt takes at least two steps, so this gives an exponential
    expected number of steps.
\end{proof}

This demonstrates that convergence can take much longer than explicit
termination, which has unfortunate consequences for detecting
convergence. If we attempt to extend the protocol above to include an
explicit termination state at some agent, the gap between the
polynomial time to reach a marked terminal configuration required by
Theorem~\ref{theorem-finish-or-fail} and the exponential time to
actually converge given by Theorem~\ref{theorem-doctor-convergence}
implies that the convergence detector will go off early in
all but an exponentially improbable fraction of executions.

\section{Examples of SWSTSs}
\label{section-examples}

When given a model of a distributed system with randomized scheduling,
we can expect to be provided with $S$ and a randomized transition rule
$→$. Turning this into an SWSTS 
requires finding a wqo $≤$ compatible
with $→$ and a weight
function $\card{⋅}$ on $S$ that gives polynomial transition
probabilities. 

Assuming at least one transition exists in any configuration,
this can be done trivially using a \concept{complete} quasi-order where
$s≤t$ for all $s$ and $t$. This is a well-quasi-order because any
infinite sequence $x_0, x_1, \dots$ has $x_0≤x1$.
For a complete wqo, compatibility with $→$ follows because for any $s→s'$
and $s≤t$, and any $t'$ with $t→t'$, $s'≤t'$.
Polynomial
transition probabilities are immediate because for any $s→s'$,
$\ProbCond{X_{i+1}≥s'}{X_i≥s} = 1$, since all $X_{i+1}≥s'$. 
We will never use this particular wqo because it makes
Theorem~\ref{theorem-finish-or-fail} trivial: every upward-closed target set
$T$ is either the empty set or the set of all configurations.
But this does suggest that we may 
need to exercise some care to make sure that the subconfiguration
relation $≤$ reflects
interesting properties of the system under consideration in addition
to satisfying the technical requirements of an SWSTS.

In this section, we show that many known models in the
literature can be described as SWSTSs with reasonably natural
subconfiguration relations. We will also give some examples
of more exotic SWSTSs that have not previously been considered.

\subsection{Population protocols}
\label{section-population-protocols}

We have been using population protocols~\cite{AngluinADFP2006} as our standard example up
until now, so let us give a brief recapitulation to illustrate the
approach:
\begin{itemize}
    \item 
        A configuration $s$ with
        $n$ agents is a vector of agent states in $Q^n$.
        Its weight function $\card{s}$ is just $n$.
        Transitions are closed because they do not change the number
        of agents.
    \item The subconfiguration relation $s≤t$ is an embedding
        $f:[\card{s}]→[\card{t}]$ mapping the indices of $s$ to the
        indices of $t$, where $f$ is strictly increasing and $s_i =
        t_{f(i)}$ for all $i$. This is a wqo by Higman's Lemma, using the
        trivial wqo on $Q$ with $q ≤ q'$ if and only if $q=q'$.
    \item Compatibility between $≤$ and $→$ follows because whenever $s→s'$
        updates the states of $s_i$ and $s_j$ to new values
        $(s'_i,s'_j) = δ(s_i,s_j)$, there is a matching transition
        $t→t'$ for any $t≥s$ that applies the same update to $t_{f(i)}
        = s_i$ and $t_{f(j)} = s_j$, giving $t'≥s'$ using the same
        embedding.
    \item Polynomial transition probabilities follow from each possible
        transition $t→t'$ occurring with probability $\frac{1}{n(n-1)} =
        Ω\parens*{\card{t}^{-2}}$, which gives a lower bound for any
        $s→s'$ on $\ProbCond{X_{i+1}≥s'}{X_i≥s}$ since at least one
        possible transition $X_i→X_{i+1}≥s'$ exists and this occurs
        with probability 
        $Ω\parens*{\card{X_i}^{-2}}$
\end{itemize}

It follows that population protocols with randomized scheduling give a
class of
closed SWSTSs. We
will refer to this class as \POP.

\subsection{Population protocols with edge oracles}

Two recent models~\cite{BlondinL2023,GanczorzGJS2023} augment the
population with additional structural information, and it is not hard
to show that these augmented models also give close SWSTSs. We will describe
these specific models first, then describe an approach that
encompasses both using an \concept{edge oracle} that carries this
additional information.

\subsubsection{Unordered and ordered models}

Blondin and Laceur \cite{BlondinL2023} define a variant of population
protocols with restricted infinite states called \concept{population
protocols with unordered data}. In this model, the state of each agent
consists of a mutable state from a finite state space $Q$ that is directly accessible to the
transition relation, and an immutable \concept{datum} or
\concept{color} from a fixed and possibly infinite
domain $\mathbb{D}$ that can only be tested for equality.
The transition relation now takes the form $δ ∈ Q^2 × \Set{=,≠} →
Q^2$, where the third input reports whether two agents have the same
color or not. We will abbreviate this model as $\POP^=$.

Ga\'{n}czorz~\etal~\cite{GanczorzGJS2023} define a similar
\concept{comparison model} where each agent has a unique hidden key; but now,
instead of testing keys for equality, keys are compared 
using a total order.
The transition relation in this model now has the form $δ ∈ Q^2 × \Set{<,>} →
Q^2$.

Because the comparison model orders all of the
agents, the input to the protocol can be treated as a word
instead of just a multiset of states. This means we can
talk about the language accepted by such a  protocol.

For symmetry
with unordered population protocols, we will refer to protocols in the comparison
model as \concept{ordered population protocols} and
abbreviate the model as $\POP^<$.

\subsubsection{Representation by edge oracles}
\label{section-representation-by-edge-oracles}

One way to generalize these models is that any interaction between
agents $x$ and $y$ is controlled by a transition rule of the form
$δ(q_x,q_y,A(x,y))=\Tuple{q'_x,q'_y}$ where $A$ is an \concept{edge oracle}
that provides information about the ordered pair $\Tuple{x,y}$ that
does not depend on the agents' states $q_x$ and $q_y$. Such an edge
oracle model is essentially a static version of the \concept{mediated
population protocols} of Michail~\etal~\cite{MichailCS2011}, which
differs by allowing the label on each edge to be updated by the
protocol. We will write $\POP$ augmented by a particular edge oracle
$A$ as $\POP^A$; the same superscript notation will later be used for other
models augmented by appropriate edge oracles.

Both $\POP^=$ and $\POP^<$ can be represented using an edge oracle,
since the oracle can provide equality testing on the agents' hidden
data or ordering information on the agents as appropriate.

\subsubsection{Orientation}
\label{section-orientation}

We have adopted the convention for $\POP$ that the initiator of an
interaction provides the first argument to an edge oracle. 

This choice is somewhat arbitrary. Most of the oracles we consider
will be \concept{reversible} in the sense that there is a function $f$
that computes $A(y,x) = f(A(x,y))$, which allows either of two agents
in an interaction to appear as the first argument to the oracle
without changing the power of the model, as long as the choice is
consistent from one interaction to another. But to be safe, we will try to adopt a
consistent ordering of the arguments to the edge oracle that applies
in other models
as well: for models with two-way communication, we will assume that one of
the agents is always marked as an initiator and appears as the first
argument to the edge oracle, while for models with one-way
communication, we will treat the agent that transmits information as
the first argument to the edge oracle.

\subsubsection{Successor oracle}
\label{section-successor}

\newData{\Successor}{successor}

Unfortunately,
population protocols will not necessarily give a useful SWSTS for an arbitrary
edge oracle. For example, suppose that the \Successor oracle for a
totally-ordered population returns true precisely when the responder
agent is the immediate successor of the initiator. This is sufficient
to implement a Turing machine simulation where each agent holds one
tape cell without error, 
since moving the tape head simply requires
waiting for the oracle to identify the immediate predecessor or
successor of the head's current position. Such an implementation can wait
exponential time by counting, violating
Theorem~\ref{theorem-finish-or-fail}. This does not necessarily mean that
there is no wqo $≤$ compatible with $→$, since as noted earlier
a complete wqo is compatible with any transition relation, but it does mean that 
any wqo $≤$ compatible with $→$ has a configuration $s$ that includes a
terminating state and a configuration $t≥s$ that does not, meaning
that termination can be lost by moving to a larger configuration.

Similar results apply to any model that can simulate a Turing machine
with sufficiently many tape cells without error. This includes
population protocols on
bounded-degree interaction graphs~\cite{AngluinACFJP2005},
community protocols~\cite{GuerraouiR2009},
mediated population protocols~\cite{MichailCS2011}, and
spatial population protocols~\cite{GasieniecKLPSS2025}, since all of
these models can, in some starting configurations,
simulate a linear Turing-machine tape by exploiting or constructing some sort of
adjacency relation between agents.

\subsubsection{Total quasi-order oracle}

However, in the case of $\POP^=$ and $\POP^<$, there \emph{are} embeddings
that give wqos.
We can prove this by defining an edge oracle model
that includes both of these models as special cases.
In this model, which we will call \concept{totally quasi-ordered
population protocols}, each agent has both a mutable state and a
hidden color from a totally ordered set $\mathbb{D}$, and the
edge oracle reports the relative order of the agents' colors, giving a
transition relation of
the form $δ ∈ Q^2 × \Set{<,=,>} → Q^2$. A
configuration with $n$ agents is a vector in $(Q×\mathbb{D})^n$, and
since the edge oracle depends only on the order of colors and not
agent positions,
for convenience we can require that the agents in any configuration are
sorted by increasing color. Since two agents may have the same color,
the total order on $\mathbb{D}$ gives a total quasi-order on agents,
hence the name. We will write the quasi-order on agents induced by the
total order on colors as $≾$, and refer to the model as $\POP^≾$.

The $\POP^=$ and $\POP^<$ models are both special cases of $\POP^≾$.
For $\POP^=$, assume a transition relation that does not distinguish
between $<$ and $>$, treating both as $≠$.
For $\POP^<$, assume that each agent starts with a unique color.

We now argue that for any protocol in $\POP^≾$, there is a wqo
$≤$ on configurations that is consistent with the transition relation
$→$ for that protocol. This will show that $\POP^≾$, $\POP^=$, and
$\POP^<$ all yield WSTSs. We start with the following lemma:

\begin{lemma}
    \label{lemma-quasi-ordered-embedding}
    Let a \concept{quasi-ordered sequence} over a set $Q$
    be a sequence $s$ of
    elements of $Q$ together with a quasi-order $≾$ on indices of $s$
    such that $i<j$ implies $i≾j$. Given two quasi-ordered sequences
    $s$ and $t$ over a well-quasi-ordered set $Q$,
    let $s≤t$ when there is a strictly increasing embedding $f$ from $s$ to
    $t$ such that $i≾j$ if and only if $f(i)≾f(j)$ and, for all $i$,
    $s_i ≤ t_{f(i)}$. Then $≤$ is a well-quasi-order.
\end{lemma}
\begin{proof}
    Write $i≈j$ for the equivalence relation that holds when $i≾j$ and
    $j≾i$, and $[i]$ for the equivalence class of $i$ under this relation.
    Observe that in any quasi-ordered sequence the elements of each
    equivalence class are consecutive. We can thus reindex $s$ as
    $s_{ci}$ where $c$ is an equivalence class and $i$ is the position
    of an element within that equivalence class.

    Write $\hat{f}(c)$ for the
    equivalence class to which $f$ maps each element of $c$. For each
    $c$, we have that $f|_c$ is an increasing embedding from $s_c$ to
    $t_{\hat{f}(c)}$
    such that $s_{ci} ≤ t_{f(c)i}$ for all $i$. By Higman's Lemma,
    such embeddings form a wqo and we have $s_c ≤ t_{f(c)}$ with
    respect to this wqo. But then $\hat{f}$ is an increasing embedding
    of the equivalence classes in $s$ to the equivalence classes in
    $t$ with $s_c ≤ t_{\hat{f}(c)}$ for all $c$, which again satisfies the
    conditions of Higman's Lemma.
\end{proof}

For $\POP^≾$, this allows us to show:
\begin{lemma}
    \label{lemma-pop-tqo}
    Every protocol in $\POP^≾$ yields a well-structured transition
    system.
\end{lemma}
\begin{proof}
    Fix some protocol in $\POP^≾$, and let $s$ and $t$ be configurations
    of this protocol. Write $s≤t$ if there is an embedding $f$ that 
    satisfies the conditions of Lemma~\ref{lemma-quasi-ordered-embedding}.
    This gives a wqo on configurations.

    Compatibility of $→$ with $≤$ follows from essentially the same
    argument as for $\POP$: given a transition $s→s'$ where agents
    $i$ and $j$ interact, and an embedding $f:s→t$,
    then there is a transition $t→t'$ in which
    agents $f(i)$ and $f(j)$ interact, observing the same states and
    relative color order as in $s$, allowing them to choose same new states
    $t'_{f(i)} = s'_i$ and $t'_{f(j)}$. It follows that $t'_{f(k)} =
    s'_k$ for all $k$, so that $f$ also gives an embedding $s'≤t'$.
\end{proof}

This shows that $\POP^≾$ (and thus also $\POP^=$ and $\POP^<$) is a
class of WSTSs. Provided $δ$ is deterministic, assuming the same
uniform random scheduling rule as for $\POP$ gives polynomial
transition probabilities, making all of these SWSTSs as well.

\subsubsection{Ancestry oracle}

Kruskal's Tree Theorem~\cite{Kruskal1960} says that finite trees with
labels from a wqo form a wqo under homeomorphic embedding, in which
$s≤t$ if there is a mapping $f$ of the vertices of $s$ to the vertices
of $t$ that is (a) order-preserving on labels in the sense that
for all $i$, the label $s_i$ of $i$ and the label $t_{f(i)}$ of $f(i)$
satisfy $s_i ≤ t_{f(i)}$; and (b) structure-preserving on the edges in
the sense that each edge in $s$ is mapped to a disjoint path in $t$.

Nash-Williams~\cite{NashWilliams1963} gives a more convenient
characterization of homeomorphisms for rooted trees.
In Nash-Williams's definition,
if $s$ and $t$ are rooted trees, $f:s→t$ is a
homeomorphism from $s$ into $t$ if, for every node $i$ in $s$, the
children of $i$ are mapped to descendants of distinct children of
$f(i)$. This has the desirable property of preserving ancestry: $f(i)$
is an ancestor of $j$ if and only if $i$ is an ancestor of $j$.
It is also easy to see that it is equivalent to part (b) of Kruskal's
more general definition, since the chains of descent mapping through
distinct children give disjoint paths in $t$ for the edges leaving
some parent in $s$.

Let us write $i ⊑ j$ if $i$ is an ancestor of $j$. This information
can be represented locally by a unique hidden string at each agent,
with $i ⊑ j$ if $i$'s string is a prefix of $j$'s string. We can
define a model $\POP^{⊑}$ where a configuration is a rooted tree
labeled by states in some finite set $Q$ and the transition function
takes the form $δ:Q^2 × \Set{⊑,⊥,⊒} → Q^2$, with the third argument
indicating whether the initiator is an ancestor of, is unrelated to,
or is a descendant of the responder.

A configuration in $\POP^{⊑}$ is now a labeled rooted tree, and an
embedding $s≤t$ is a label-preserving homeomorphism between rooted
trees.  Configurations in $\POP^{⊑}$ form a wqo by Kruskal's Tree
Theorem (with, as usual, the trivial wqo on $Q$). Compatibility with
$→$ follows from $⊑$ being preserved by embedding. Polynomial
transition probabilities follow from all transitions occurring with
the same $\frac{1}{n(n-1)}$ probability as in all the other $\POP$
variants. This makes $\POP^{⊑}$ is a class of closed SWSTSs.

\subsection{Chemical reaction networks}
\label{section-chemical-reaction-networks}

The population protocol model is a special case of
the \concept{chemical reaction network} (CRN)
model~\cite{Gillespie1977}.
In addition to 
pairwise interactions of the form $A+B → C+D$, a chemical reaction
network may have
reactions that involve any positive number of reactants on the
left-hand side and products on the right-hand side.
The intuition is that in this model, agents correspond to
molecules, and steps correspond to chemical reactions.
A chemical reaction network is specified by giving a
list of permitted reactions, such as the reaction $H_2 + H_2 +
O_2 → H_2O + H_2O$ that converts two molecules of $H_2$ and
one molecule of $O_2$ into two molecules of $H_2 O$.
The types of molecules that may appear in a reaction are called
\concept{species} and correspond to the states of agents in other
models.
Each possible reaction occurs at a \concept{rate} that is some
constant times the number of possible ways to supply the reagents. The
probability of a particular reaction occurring at a given step is the
ratio between its rate and the sum of the rates of all possible reactions.

Because CRNs do not preserve the number of agents, 
we cannot use the $Q^n$ representation used for population protocols.
Instead, we
model a configuration of a CRN as a
finite-dimensional vector in $ℕ^Q$ of counts of molecules of each species.
The subconfiguration relation is componentwise ordering: $s≤t$ if and
only if $s_q ≤ t_q$ holds for all $q∈Q$.
A reaction is specified by a pair of vectors
$\Tuple{x,y}$,
where $s→s-x+y$ whenever $x≤s$.

Componentwise ordering is a wqo by Dickson's
Lemma. That it is compatible with $→$ can easily be shown by observing
that $s→s'=s-x+y$ and $s≤t$ implies $x≤s≤t$, giving a transition
$t→t'=t-x-y$ with $s'≤t'$.

The weight of a configuration is obtained by assigning a total atomic
number $w_q$ to each species $q∈Q$, and computing $\card{s} = ∑_q w_q
s_q$. This is preserved for realistic chemical reactions that include
all reagents and products, but not necessarily for arbitrary chemical
reaction networks that may allow generation of new products, perhaps
by consuming feedstocks that are not explicitly included in the
reagents.

It remains to show that this model gives polynomial transition probabilities. 
Any transition $s→s'$ has $s'
= s - x + y$ for some reaction $\Tuple{x,y}$. The rate $r_{xy}$ of this
reaction in $t≥s$ is given by $c_{xy} ∏_{q∈Q} (t_q)_{x_q}$ where
$c_{xy}$ is the base rate for the reaction and
$(t_q)_{x_q} = t_q (t_q - 1) \dots (t_q - x_q + 1)$ counts the number
of ways of supply $x_q$ copies of molecule $q$ in configuration $t$.
Each term in this product is at least $1$ and at most
$\card{t}^{x_q}$, assuming no species has atomic number less than $1$.
It follows that $r_{xy}$ is polynomial in $\card{t}$, and that the
probability $\frac{r_{xy}}{∑_{x'y'} r_{x'y'}}$ that this is the next
reaction to occur is also polynomial in $\card{t}$.

We thus have another class of SWSTSs, which we will refer to as \CRN.

\subsubsection{Non-closed CRNs}

As noted, SWSTSs in this class may or may not be closed.
Closed CRNs satisfy the conditions of Theorem~\ref{theorem-finish-or-fail}
and thus have the same
property of finishing or failing in polynomial expected time as other
closed SWSTSs. But non-closed CRNs can be used to demonstrate that
these results do not necessarily hold for SWSTSs that are not closed.

The example below is based on a construction of
Soloveichik~\etal~\cite{SoloveichikCWB2008}.
Consider the CRN with species $a$, $b$, and $c$, and transitions
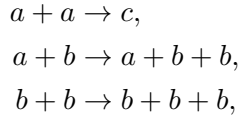
\begin{align*}
    a + a &→ c,\\
    a + b &→ a + b + b,\\
    b + b &→ b + b + b,
\end{align*}
all of which occur at the same constant rate.
Let $T$ be the
upward-closed set of configurations that contains at least one $c$.
For each $i∈ℕ$, let $s_i$ consist of two molecules of species $a$ and $i$
molecules of species $b$. From $s_i$, there are $2i + i(i-1) = i(i+1)$ possible
transitions to $s_{i+1}$ (by applying either the $a+b$ or the $b+b$
rules) and one possible transition to $T$ (by applying the $a+a$
rule). The transition to $T$ occurs with probability
$\frac{1}{i(i+1)+1)} < \frac{1}{i^{-2}}$. Start an execution in $s_2$,
and let $X_i$ be the indicator for the event that
$s_i$ transitions to $T$. Then $\Exp{X_i} < i^{-2}$ and $\Exp{∑ X_i} <
\frac{π^2}{6} - 1 < 1$. So the union bound gives a nonzero probability
that the execution never reaches $T$ or a state that cannot reach $T$.
This contradicts the conclusion of Theorem~\ref{theorem-finish-or-fail},
demonstrating that some bound on the weight is necessary.

\subsubsection{CRNs and edge oracles}

Because the number of agents in $\CRN$ may change over time, it is not
obvious how to apply an edge oracle to the $\CRN$ model. For some
special cases it may be possible to define an extension to $\CRN$ in
which output molecules inherit hidden attributes of specific input
molecules, but the effort of doing this may be more trouble than it is
worth. So we will restrict our attention to the base $\CRN$ model
without oracles.

\subsection{Synchronous gossip models}

Population protocols and chemical reaction networks are examples
of asynchronous models, where only a single interaction or reaction
happens at a time. Such models naturally yield well-structured
transition systems, because when $s≤t$ means that $s$ corresponds to a
subcollection of agents in $t$, we can map any transition $s→s'$ to a
transition $t→t'$ by having any extra agents in $t$ do nothing.

This is not the case in a synchronous model in which every agent
updates its state in every step. Fortunately, we can often still get a
WSTS by having the agents corresponding to $s$ in a
configuration $t≥s$ interact in a way that doesn't depend on the
states of any extra agents in $t$.

In this section, we describe several models motivated by classic
gossip algorithms designed for rumor spreading
(see~\cite{HedetniemiHL1988} for a survey of early literature in this
area). Gossip algorithms come in
many variants, so we do not claim that these models necessarily
capture every aspect of these algorithms; instead, we take inspiration
from previous examples of generic gossip models,
including the very general sampling-based \concept{anonymous gossip
protocol} model of
Bertier~\etal~\cite{BertierBK2009} and the synchronous
uniform $\GOSSIP$ model of Bechetti~\etal~\cite{BechettiCNPS2014},
which formalize some of the underlying assumptions of classic
rumor-spreading algorithms in a way that allows for more generalized
computations.

We assume finite-state agents throughout, and as with $\POP$, define a
configuration as a sequence in $Q^*$ for some finite state space $Q$.
The subconfiguration relation $≤$ is ordered embedding, giving a wqo.
When adding edge oracles, the subconfiguration relation will in
addition require preserving the outputs given by the edge oracle;
showing that such a restricted $≤$ works for a particular edge oracle
will follow the same arguments as for $\POP$ and will be omitted here.

Most of these models involve one-way transmission of information. As
discussed in §\ref{section-orientation}, we will assume that when an
edge oracle is supplied, the transmitter of this information will be
supplied as the first argument to the oracle.

\subsubsection{$\GOSSIP$}

In the Bechetti~\etal~model,
at each step, each of $n$ agents samples another
agent then updates it own state according to a transition
function $δ:Q×Q→Q$ that is applied to the state of the agent
itself and that of the agent that it samples. 
Using the terminology of e.g.~\cite{KarpSSV2000,KempeDG2003},
this is a \concept{pull} model, in that
each agent pulls information from the agent it samples. 

As in $\POP$, configurations in this model can be represented
by vectors of agent states, with $≤$ given by ordered
embedding, which makes $≤$ a wqo when $Q$ is finite.

The transition relation $→$ is compatible with $≤$, because
if $s→s'$ and $s≤t$, there is a transition $t→t'$ with $s'≤t'$
in which all agents in $s$
sample only other agents in $s$.

Polynomial transition probabilities follow from calculating, given a
transition $s→s'$ and an inclusion $s≤t$, the probability that the
agents in $s$ sample the same other agents in both $s$ and $t$; this
occurs with probability exactly $\card{t}^{-\card{s}}$, which is
polynomial in $\card{t}$ since $\card{s}$ is a constant for any fixed
transition $s→s'$.

Following Bechetti~\etal, we will refer to this model as $\GOSSIP$.

\newcommand{\Oracles}[1]{${#1}^{=}$, ${#1}^{<}$, ${#1}^{≾}$, and ${#1}^{⊑}$}

As with $\POP$, we can extend $\GOSSIP$ with an edge oracle, which in
this case tells an agent updating its state what its relation is to
the agent it samples. Specifically, when $x$ samples $y$, it computes
a new state $δ(q_x,q_y,A(y,x))$, where ordering of the arguments to
$A(y,x)$ follows the convention that the transmitter of information
comes first.
Choosing the appropriate wqo for each oracle makes 
\Oracles{\GOSSIP}
all classes of SWSTSs.

\subsubsection{$\PUSH$}

Again using the terminology of~\cite{KarpSSV2000,KempeDG2003},
the pull model used in $\GOSSIP$ contrasts with
a \concept{push} model, in which information flows in the other
direction, from an agent to the agent it samples.
Here at each step, each agent updates its
state according to a transition function $δ:Q×ℕ^Q→Q$, where
the input to $δ$ is the state of the agent and the
unordered multiset of states that are pushed to it.
We will refer to this model as $\PUSH$.

The $\PUSH$ model with the ordered-embedding subconfiguration relation
is a WSTS, because if $s→s'$ and $s≤t$, then there is a
transition $t→t'$ with $s'≤t'$ in which each agent in $s$
receives only pushes from agents in $s$.

Showing polynomial transition probabilities is a bit trickier than in
$\GOSSIP$. Given a
transition $s→s'$ and embedding $f:s→t$, it is not enough for every
agent in $s$ to push to the corresponding agent in both $s$ and $f(s)$; we also
need that any new agents in $t∖f(s)$ do not push to agents in
$f(s)$. Write $n = \card{t}$ and $m= \card{s}$. 
Then for each of the $m$ agents $x∈s$ where $x$ pushes to $y∈s$ in the $s→s'$
transition, there is a probability of exactly $1/n$ that $f(x)$ pushes
to $f(y)$, and for each of the $n-m$ agents $z∈t∖f(s)$, there is a probability of exactly
$1-m/n$ that $z$ pushes to another agent in $t∖f(s)$.
These events are all independent, so the probability that they all
occur is exactly
\begin{align*}
    n^{-m} \parens*{1-\frac{m}{n}}^{n-m}
    &≥
    n^{-m} \parens*{1-\frac{m}{n}}^{n}
    \\&=
    Θ\parens*{n^{-m}},
\end{align*}
since the right-hand factor converges to $e^{-m} = Θ(1)$ for large
$n$.

Given an edge oracle $A$ with output values $V$, $\PUSH^A$ replaces
the transition function with a function $δ:Q×ℕ^{Q×V}→Q$.
For an agent $x$, the counts now track both the state of each agent
$y$ that pushes to $x$ as well as the output of the edge oracle
$A(y,x)$. The same wqos used for
$\POP$ produce stochastic WSTSs for all of \Oracles{\PUSH}.

One issue with $\PUSH$ and its edge-oracle variants is that the
transition function requires unbounded space to store, since the
multiset of observed states can be arbitrarily large. For some
applications it will be helpful to impose a uniformity condition, for
example that the transition function is computable using space
logarithmic in the number of agents. An alternative (which we do not
explore) might be use a constant-sized transition function that
processes the incoming states one at a time in an arbitrary order,
possibly with an extra marker to indicate when the states obtained in
a particular step are complete.

\subsubsection{Matching model}

Here the size of the population is
assumed to be even, and at each step, a directed matching is
applied to the population and each pair of matched agents
updates their states according to a transition function
$δ:Q^2→Q^2$. 

By a \concept{directed matching} we mean a partition of the set of agents
into pairs, with a direction assigned to each pair. A straightforward way
to sample such matching is simply to pick one of the $n!$ possible
permutations of the agents, split the resulting sequence into pairs,
and then forget the ordering of the pairs. This yields
$\frac{n!}{(n/2)!}$ possible ordered matchings, all of
which occur with equal probability.

As usual we let $s≤t$ when there is an ordered embedding $f:s→t$
that preserves states.
Such protocols form WSTSs because when $s→s'$ and
$s≤t$, there is at least one transition $t→t'$ with $s'≤t'$
obtained by having all the agents in $f(s)$ match with
each other in the same pattern as in $s$.

When $\card{t} = n$ and $\card{s} = m = O(1)$, for any fixed directed
matching on $s$ giving a particular transition $s→s'$,
we can construct a directed matching on $t$ that
applies the same directed matching to $f(s)$ by choosing a directed
matching for all the agents not in $f(s)$ in one of
$\frac{(n-m)!}{\parens*{(n-m)/2}!}$ possible ways. This gives a
probability of reaching a configuration $t'≥s'$ from $t$ of at
least
\begin{align*}
    \frac{(n-m)!/\parens*{(n-m)/2}!}{n!/(n/2)!}
    &= \frac{(n/2)!/\parens*{(n-m)/2}!}{n!/(n-m)!}
    \\&= \frac{Θ\parens*{(n/2)^{m/2}}}{Θ(n^m)}
    \\&= Θ\parens*{n^{-m/2}},
\end{align*}
which is polynomial in $n$ for fixed $m$.

We will refer to this model as $\MATCHING$. Adding edge oracles gives
the usual SWSTS variants \Oracles{\MATCHING}. As with other models
with two-way communication, the convention for
applying each oracle is that the first agent in each pair supplies the
first argument to the oracle.

\subsubsection{Token shuffling model}

Alistarh~\etal~\cite{AlistarhGR2022} define a synchronous gossip-like
model in which, for some fixed $k≥1$,  each of $n$ agents generates $k$ tokens based on its current
state, all $nk$ tokens are shuffled uniformly among the agents, and
then each agent updates its state based on its previous state and the
tokens it received according to a transition function $δ:Q×Q^k→Q$.
This is similar to the $\PUSH$ model, except that it is guaranteed that
each agent receives tokens from exactly $k$ agents, and it is possible
that the agent receives tokens from itself.

As in other models, set $s≤t$ when there is an ordered embedding $f$ from $s$ into
$t$. Given a transition $s→s'$ and embedding $f:s→t$, there is at
least one compatible transition $t→t'$ in which, for each $x∈s$, $f(x)$ receives
the same tokens as $x$. We would like to show that such transitions
occur with probability polynomial in $n = \card{t}$ for any fixed
transition $s→s'$. As before we will let $m=\card{s}$.

To show polynomial transition probabilities,
observe that the total number of possible shuffles of $nk$ tokens into
$n$ ordered sequences of $k$ tokens each
is exactly $(nk)!$, since we can just pick a permutation
and then chop it up into segments of size $k$.
To get a transition that corresponds to $s→s'$, we want to (a) send all the
tokens from agents in $f(s)$ to other agents in $f(s)$ in a way that
reflects the $s→s'$ transition, and (b) send the remaining $(n-m)k$
tokens arbitrarily to the remaining $n-m$ agents. There is at least
one way to do the first part and there are exactly $((n-m)k)!$ ways to do the
second, giving
\begin{align*}
    \ProbCond{X_{i+1}≥s'}{X_i = t ≥ s}
    &≥ \frac{((n-m)k)!}{(nk)!}
    \\&= \frac{1}{(nk)_{mk}}
    \\&≥ (nk)^{-mk}
    \\&= Θ\parens*{n^{-m}}.
\end{align*}

We will refer to this class of SWSTSs as \SHUFFLE, and define
corresponding edge oracle classes \Oracles{\SHUFFLE}, with the
convention that a token sent to $x$ by $y$ is marked with $A(x,y)$. We
will also write $\SHUFFLE_k$ for the subclass of $\SHUFFLE$ consisting
of protocols where each agent emits $k$ tokens.

\subsection{A model based on graph immersion}

Because well-quasi-orders are central to the definition of a WSTS, a
straightforward way to find more WSTS models is to consider other
classes of well-quasi-orders. The idea is to use a known wqo to impose a
quasi-order on configurations, then look for a transition rule that is
compatible with this quasi-order. In this section, we give an example
of this approach, based on a particular wqo defined on graphs.

The Graph Minor Theorem of Robertson and
Seymour~\cite{RobertsonS2004} shows that graphs form a wqo under the
operation of taking minors, where $G$ is a minor of $H$ if $G$ can be
obtained by contracting edges of a subgraph of $H$. Unfortunately, it
is difficult to come up with a plausible transition rule that is
compatible with this wqo, mostly because it is not clear what to do
with the states of the nodes that are merged by edge contractions.

However, a later result by the same authors~\cite[Theorem 1.5]{RobertsonS2010} shows
that graphs form a wqo under \concept{immersion}, which maps vertices
of $G$ to vertices of $H$ and maps edges of $G$ to disjoint paths in
$H$, while preserving a wqo on the labels of the vertices (which, as
usual, we will limit to equality over a finite state set $Q$). Immersion
preserves \concept{betweenness}: if $f:G→H$ is an immersion and $t$
appears as an intermediate vertex in a simple path from $s$ to $u$,
then $f(t)$ appears as an intermediate vertex on a simple path from $f(s)$ to $f(u)$.

We can use this to get a class of SWSTSs that we will call
$\BETWEEN$. In this model, each transition consists of choosing a
triple $\Tuple{s,t,u}$ of vertices, such that $t$ is on a simple path
from $s$ to $u$, with a probability that is uniform among all such
triples. The agents then update their state according to a transition
function $δ:Q^3→Q^3$. Since any such triple $\Tuple{s,t,u}$ maps 
through any immersion $f:G→H$ to a corresponding triple
$\Tuple{f(s),f(t),f(u)}$ of agents with the same states that also
satisfy the betweenness condition, any transition $G→G'$ maps to a
corresponding transition $H→H'$ where $f(s), f(t)$, and $f(u)$ update
their states in $H$ according to the same transition that $s,t$, and
$u$ apply in $G$. We also have polynomial transition probabilities,
because there are at most $n(n-1)(n-2)$ possible transitions in a graph with
$n$ agents, giving a minimum transition probability that is
$Ω\parens*{n^{-3}}$.

It is important that $\BETWEEN$ is constructed by limiting which
agents can interact rather than supplying an extra argument telling an
arbitrary triple of agents how they are related to each other. This is
because immersion does not preserve \emph{non}-betweenness: by adding edges
to a graph $G$, we can create a new path that puts $t$ between $s$ and
$u$ in some $H≥G$ even though $t$ is not between $s$ and $u$ in $G$.
But if we only allow interactions for triples $\Tuple{s,t,u}$ where
$t$ is between $s$ and $u$, adding more such triples just adds more
possible interactions, which is permitted by compatibility between $≤$
and $→$. The same issue complicates extending $\BETWEEN$ using an edge
oracle, so we do not consider oracle variants of this model.

We cannot deny that the $\BETWEEN$ model is an artificial one,
motivated more by a desire to use a particular wqo than any claim that
it represents a realistic model of computation. But it does illustrate
how even unusual wqos can produce SWSTSs.

\section{Computation by SWSTSs}
\label{section-computation}

To use an SWSTS as a model of computation, we must be able to
encode inputs and decode outputs. This leads to the following extended
definition:
\begin{definition}
    \label{definition-SWSTS-protocol}
    An \concept{SWSTS protocol}
    $(S,≤,→,\card{⋅},\Pr,Σ,I,V_0,V_1)$
    consists of:
    \begin{itemize}
        \item An SWSTS $(S,≤,→,\card{⋅},\Pr)$.
        \item A finite \concept{input alphabet} $Σ$.
        \item An \concept{input map} $I:Σ^*→S$ such that for any $x$, $\card{I(x)}$ is polynomial in
            $\card{x}$, and for any $x$ and $y$ where $x≤y$ in the
            subsequence ordering, $I(x)≤I(y)$.
        \item Two \concept{output sets} $V_0$ and
            $V_1$, that are closed under both $≤$ and $→$,
            representing terminal configurations with outputs $0$ and $1$, 
            respectively.
    \end{itemize}
\end{definition}

We can now compute a Boolean function $f:Σ^*→\Set{0,1}$ on an input
$x$ by starting in configuration $I(x)$ and running the SWSTS until we
reach a configuration in $V_0$ or $V_1$. Typically this will involve a
small chance of error (see Section~\ref{section-errors}), so we will
say that a SWSTS protocol $P$ computes $f$ with error $ε$ if,
for each $x$, $P$ outputs the correct value of $f(x)$ with probability
at least $1-ε$. Similarly we can accept a language $L$ with error $ε$
if, for each $x$, we have a probability of at least $1-ε$ of reaching
$V_0$ if $x∉L$ or $V_1$ if $x∈L$.

For example, a population protocol might
compute some function $f$ by having $I$ map the input letters to the input
states of a sequence of agents,
while adding a leader agent and possibly some extra agents
supplying additional storage, then running a register machine
simulation as in~\cite{AngluinAE2008fast} until the leader reaches
a terminating state with embedded output $0$ or $1$. The sets $V_0$
and $V_1$ will consist of all configurations that contain at least one
leader agent with output $0$ or $1$. Applying the error analysis
in~\cite{AngluinAE2008fast}, the probability $ε$ of error 
can be made polynomially small.

In analyzing SWSTS protocols, we will want to apply
Lemma~\ref{lemma-polynomial-path} and
Theorem~\ref{theorem-finish-or-fail}, but these only apply to closed
SWSTSs. A \concept{closed SWSTS protocol} will be one whose SWSTS is
closed.

Definition~\ref{definition-SWSTS-protocol} imposes several constraints
on the input map and output sets.

For the input map, the requirement to preserve the input size up to
polynomials means that polynomial transition probabilities for closed
SWSTS protocols become polynomial in the size of the input. The
requirement that the input map preserves the subsequence order is a
bit more artificial, but it intuitively reflects the idea that the
initial configuration should reflect the structure of the input as
much as possible, and allows for some statements about initial
configurations back into statements about inputs (see, for example,
Theorem~\ref{theorem-errors-everywhere}).

The constraints that we can express on the input map
are limited by the level of abstraction of the set of configurations
$S$. For example, we might like to insist that $I$ can't do our
computations for us by limiting $I$ to, say, a finite-state or
log-space transducer, but since $S$ is not necessarily a set of
strings, we have to allow for more general input maps. It might be
hoped that the requirement that $I$ preserve the subsequence order
would help here, but in general this is not the case: given some
arbitrary Boolean function $f:Σ^*→\Set{0,1}$, the input map $I$ that
sends $x$ to $0^{2\card{x}+f(x)}$ satisfies both the polynomial
expansion and order constraints, and allows any SWSTS that is clever
enough to compute parity to compute $f$. So instead we will have to
exercise restraint and avoid unreasonable choices of input maps.

The constraints on the output sets are intended to reflect the idea
that an output can be obtained by examining a small part of the
configuration and so is not changed by considering more components or
continuing to run the protocol past termination. For
SWSTS protocols involving agents, each output set $V_i$ will typically
consist of all configurations containing some special agent marked
with output $i$.

However, insisting that $V_0$ and $V_1$
are upward-closed means that we now cannot in general demand that they
are disjoint. For example, a configuration of a population protocol
that somehow acquired an extra leader agent might have its two leader
agents decide different values. Usually all we can ask is that any
configuration that is reachable from an initial configuration appears
in at most one of $V_0$ and $V_1$.

\subsection{Errors}
\label{section-errors}

Like most probabilistic computational models, SWSTS protocols are
vulnerable to error. The generality of the SWSTS approach lets us
prove this without having to reference the details of any specific model.

\begin{theorem}
    \label{theorem-errors-everywhere}
    Let $P = (S,≤,→,\card{⋅},\Pr,Σ,I,V_0,V_1)$ be a closed SWSTS
    protocol. Then either $P$ outputs the same value in all
    executions, or there is a non-empty, upward-closed set $A$ such
    that for every input $x∈A$, $P$ computes the wrong output for
    $x$ with at least polynomial probability.
\end{theorem}
\begin{proof}
    Let $A = I^{-1}\parens*{\Pred^*{V_0}∩\Pred^*(V_1)}$. Then for any
    $x∈A$ and $i∈\Set{0,1}$, $I(x)∈\Pred^*(V_i)$ and by
    Lemma~\ref{lemma-polynomial-path}, an execution starting in $I(x)$
    reaches $V_i$ with at least polynomial probability.
    Since this is true for both $V_0$ and $V_1$, and only one of these
    outputs is correct for $x$, for any $x∈A$, there is a polynomial
    probability that $P$ computes the wrong output for $x$.

    From Lemma~\ref{lemma-predecessors}, since each output set $V_i$
    is upward-closed, $\Pred^*(V_i)$ is also upward-closed.
    So $\Pred^*(V_0)∩\Pred^*(V_1)$, as the intersection of two
    upward-closed sets, is upward-closed, and since $I$ is
    order-preserving, $A=I^{-1}\parens*{\Pred^*{V_0}∩\Pred^*(V_1)}$ is
    upward-closed.

    Suppose now that there exist inputs $x_0$ and $x_1$ that yield
    outputs $0$ and $1$ in some executions. Then $I(x_0)∈\Pred^*(V_0)$ 
    and $I(x_1)∈\Pred^*(V_1)$. Let $y$ be the concatenation $x_0x_1$
    of $x_0$ and $x_1$. Then $x_0,x_1≤y$ in the subsequence order so
    $I(x_0),I(x_1)≤I(y)$. Because each $\Pred^*(V_i)$ is
    upward-closed, it follows that $I(y)$ is in both $\Pred^*(V_0)$
    and $\Pred^*(V_1)$, which puts $y$ in $A$. So either $P$ outputs
    the same value in all executions, or $A$ is nonempty.
\end{proof}

Theorem~\ref{theorem-errors-everywhere} shows that non-constant SWSTS
protocols have a polynomial probability of error for almost every
input that is sufficiently large. This is because if $x$ is some element of a
non-empty set of strings $A$ over an alphabet $Σ$ that is upward-closed
with respect to the subsequence relation, then the expected number of
random draws from $Σ$ until we get a supersequence of $x$ is only the
constant value $\card{x}⋅\card{Σ}$, and by Markov's inequality this
means that a random string $y$ of length $n$ will have $x≤y$ (and thus
be in $A$) with probability at least
$1-\frac{\card{x}⋅{\card{y}}}{n}$, which converges to $1$ in the limit. This
provides some justification for assuming that computations by
SWSTS protocols must allow for two-sided error.

\section{Simulating SWSTS protocols with Turing machines}
\label{section-simulating-SWSTS}

An SWSTS protocol terminates when it reaches a configuration in $V_0$
or $V_1$. Because these sets are required to be upward-closed,
Theorem~\ref{theorem-finish-or-fail} applies: after a polynomial
expected number of steps, a closed SWSTS protocol either terminates or
reaches a configuration from which it is no longer possible to
terminate. We will use this fact to show that a probabilistic
polynomial-time Turing machine can compute the output of an SWSTS
protocol with polynomially small error, as long as it can represent
configurations efficiently.

We give two versions of this result.
Section~\ref{section-simulating-SWSTS-in-BPP} gives sufficient
conditions for simulating an SWSTS protocol in probabilistic
polynomial time, while Section~\ref{section-simulating-SWSTS-in-BPL}
gives stronger conditions that allow a protocol to be simulated in
probabilistic log space.

\subsection{Simulating SWSTS protocols in $\classBPP$}
\label{section-simulating-SWSTS-in-BPP}

Call an SWSTS \concept{uniform} if there is a mapping
from configurations $s$ to representations $\hat{s}∈∑^*$ for some
finite alphabet $Σ$, such that
(a) $\card{\hat{s}}$ is polynomial in
$\card{s}$;
(b) there is a probabilistic polynomial-time
Turing machine that takes as input a representation $\hat{s}$ of a
configuration $s$ and outputs $\hat{t}$ for each $s→t$ with
probability at least $(1-o(n^{-c}))\Prob{s→t}$ for any constant $c>0$;
and
(c) for any fixed $s$, there is a Turing
machine that computes $s≤t$ given $\hat{t}$ in time polynomial in
$\card{\hat{t}}$.

Call an SWSTS protocol uniform if, in addition to the above
requirements, there is a polynomial-time Turing machine that takes
input $x$ and produces output $\widehat{I(x)}$. 

\begin{theorem}
    \label{theorem-uniform-SWSTS-in-BPP}
    Let $L$ be a language accepted by a uniform closed SWSTS protocol
    with bounded two-sided error $ε<1/2$. Then $L$ is in $\classBPP$.
\end{theorem}
\begin{proof}
    Let $P = \Tuple{S,→,≤,\card{⋅},\Pr,I,V_0,V_1}$ be a protocol that accepts
    $L$ with error bound $ε<1/2$.

    From Theorem~\ref{theorem-finish-or-fail}, there exists a constant
    $k$ such that an execution of $P$ on input $x$ either reaches a
    configuration in $V_0$ or $V_1$, or reaches a configuration that
    cannot reach $V_0$ or $V_1$, in expected
    $O\parens*{\card{I(x)}^k}$ steps, which is polynomial in $n =
    \card{x}$. Choose a constant $c$ so that this polynomial is $O(n^c)$.

    Observe that because $V_0$ and $V_1$ are both upward-closed, we
    can compute membership of some configuration $t$ in either set in
    polynomial time from $\hat{s}$ by checking whether $t$ lies above
    any of a finite number of basis elements $s$.

    Given an input $x$, simulate a execution of $P$ of length at most
    $n^{c+1}$ in polynomial time by computing
    $X_0 = \widehat{I(x)}$, then repeatedly sampling transitions $X_i
    → X_{i+1}$ until either (a) the simulation reaches some
    configuration in $V_0$ or $V_1$, or (b) the number of simulated
    steps exceeds $n^{c+1}$. In the first case, output the appropriate
    value $0$ or $1$; in the second case, output $0$. The second case
    may introduce an additional source of error by terminating
    the computation prematurely, but Markov's inequality says that the
    probability that it occurs is $O(1/n)$. 

    To bound the probability that the simulated protocol produces the
    correct answer, observe that for any particular execution $s_0 → s_1 \dots
    s_m$ of the SWSTS protocol that occurs with some probability
    $p = ∏_{i=0}^{m-1} \Prob{s_i→s_{i+1}}$,
    each transition $\hat{s}_i → \hat{s}_{i+1}$ is chosen in the simulation 
    with probability at least $1-o(n^{-c-2}) \Prob{s_i → s_{i+1}}$,
    giving a probability of sampling this execution of at least
    $(1-o(n^{-c-2}))^t ∏_{i=0}^{m-1} \Prob{s_i → s_{i+1}} ≥
    (1-o(n^{-c-2}))^{n^{c+1}} p = (1-o(1)) p$, since the error factor
    converges to $1$ in the limit. Aggregating these
    errors over all accepting/rejecting executions adds at most $o(1)$
    to our error budget.

    So the total error is bounded by $ε+O(1/n)+o(1) < 1/2-Ω(1)$ 
    for sufficiently large $n$, putting $L$ in $\classBPP$.
\end{proof}

For example, any $\POP^<$ protocol with a polynomial-time-computable
input map $I$ is uniform since we are already
representing configurations as sequences of agent states
drawn from a finite set, giving (a) $\card{\hat{s}} =
\card{s}$; (b) polynomial-time sampling of $s→t$ since we can just
pick a pair of agents from the sequence\footnote{Doing this exactly in
polynomial time assuming our probabilistic Turing machine only has
access to random bits cannot be guaranteed unless $n$ and $n-1$ are
both powers of $2$; but this is accounted for by allowing an exponentially small
relative error $1-o(n^{-c})$ in the definition of uniformity.} and update their states
according to the constant-time-computable transition function $δ$; and (c)
polynomial-time testing if $s≤t$ since we can test
if a fixed $s$ is a subsequence of $t$ using a finite-state machine.
So Theorem~\ref{theorem-finish-or-fail} tells us that any
language accepted by a $\POP^<$ protocol with bounded two-sided error
and a polynomial-time-computable input map fits in $\classBPP$.

A similar analysis applies to the base $\POP$ model and any oracle
variant $\POP^A$ that gives an SWSTS. For closed $\CRN$ protocols with
rational reaction rates and polynomial-time-computable input maps,
standard simulation algorithms~\cite{Gillespie1977,GibsonB2000}
demonstrate uniformity. Similarly, uniformity holds for $\GOSSIP$,
$\MATCHING$, $\SHUFFLE$, and (assuming $δ$ is polynomial-time
computable) $\PULL$, as well as their oracle variants. So all of these
classes of protocols can be simulated in $\classBPP$.

It is not clear whether $\BETWEEN$ yields uniform SWSTSs in general.
The issue is that while it is known that testing graph immersion is
fixed-parameter tractable~\cite{GroheKMW2011}, it is not immediately
obvious if this result extends to immersions of labeled graphs.
So it may not be the case that the simulation in the proof of
Theorem~\ref{theorem-uniform-SWSTS-in-BPP} can efficiently detect
membership in arbitrary upward-closed sets $V_0$ and $V_1$.
We can avoid this issue for specific $\BETWEEN$ protocols where $V_0$
and $V_1$ only test for the presence of an agent in an appropriate
output state, which is easily checked by a polynomial-time TM.
For protocols using this output convention, we get languages in $\classBPP$.

\subsection{Simulating symmetric SWSTS protocols in $\classBPL$}
\label{section-simulating-SWSTS-in-BPL}

In some cases, we can take advantage of symmetry to simulate an SWSTS
in $\classBPL$ instead of $\classBPP$. As with $\classBPP$, we rely on
Theorem~\ref{theorem-finish-or-fail} to guarantee termination in
polynomial time, but rely on additional specific properties of the
SWSTS to show that it can be simulated in logarithmic space.
The key idea is that if the agents in an SWSTS protocol are
interchangeable, then we can replace a polynomial-sized 
representation of a configuration as a string in $Σ^*$ with a
logarithmic-sized representation as a multiset in $ℕ^Σ$.

Call an SWSTS \concept{symmetric} if there is an equivalence relation
$\sim$ on configurations such that if $s \sim s'$, then for any $t$,
$\sum_{t' \sim t} \Prob{s→t'} = \sum_{t' \sim t} \Prob{s'→t'}$.
Letting $[s]$ be the equivalence class containing $s$, this says that
there is a well-defined probability $\Prob{[s]→[t]} = ∑_{t' \sim t}
\Prob{s→t'}$ of a transition to an element of $[t]$ from an element of
$[s]$ that does not depend on the choice of representative $s$.

Call a closed SWSTS \concept{symmetrically log-space-uniform} if
(a) it is symmetric;
(b) there is a mapping from configurations $s$ to representations
$\hat{s} ∈ Σ^*$ for some finite alphabet $Σ$ that preserves symmetry
in the sense that if $s \sim t$, then $\hat{s}$ is a permutation of
$\hat{t}$;
(c) there is a probabilistic log-space transducer that takes as input a
representation $\hat{s}$ of a configuration $s$ and outputs a
configuration $\hat{t}$, with the property that the probability of
outputting some $t'$ in $[t]$ is at least $(1-o(n^{-c}))
\Prob{[s]→[t]}$ for any constant $c>0$; and
(d) for any fixed $s$, there is a Turing machine that computes $s≤t$
given $\hat{t}$ using space logarithmic in $\card{\hat{t}}$.

Call a closed SWSTS protocol symmetrically log-space-uniform if, in
addition to the above requirements, there is a log-space transducer
that takes input $x$ and produces an output that is a permutation of
$\widehat{I(x)}$.

\begin{theorem}
    \label{theorem-log-space-uniform-SWSTS-in-BPL}
    Any language accepted by a symmetrically log-space-uniform SWSTS
    protocol is symmetric and contained in $\classBPL$.
\end{theorem}
\begin{proof}
    Symmetry is immediate from the fact that for any computation by
    the protocol on some input $x$ that reaches $V_0$ or $V_1$,
    there is an equivalent computation on any permutation of $x$ that
    reaches the same set with the same probability.

    Proving containment in $\classBPL$ is essentially the same as for
    Theorem~\ref{theorem-log-space-uniform-SWSTS-in-BPL}, except that
    we replace polynomial-time steps throughout with log-space steps,
    and must take some care about representing intermediate
    configurations to avoid using more than logarithmic space.

    Because $\Prob{[s]→[t]}$ doesn't depend on which members of the
    equivalence classes we pick, we can represent $[\hat{s}]$ as an
    element $\widehat{[s]}$ of $ℕ^Σ$ by counting how many times each
    symbol appears in $s$. Given such counts (which take space
    logarithmic in $\card{\hat{s}}$), we can simulate the input to the
    log-space transducer sampling $→$ by choosing a sorted permutation
    of $s$ and calculating the symbol at each simulated input tape
    location from the counts in $\widehat{[s]}$. Similarly, each time
    the transducer write a symbol to its output tape, we replace this
    operation with increment a counter to eventually obtain a vector
    of counts $\widehat{[t]}$.

    Applying the input map $I$ and testing for the terminal
    test $V_0$ and $V_1$ uses the same technique to represent
    equivalence classes of configurations. The symmetry properties of
    the SWSTS mean that replacing each representation with its sorted
    version does not affect the behavior of any of these parts of the
    simulation.

    Error analysis is identical to that in the proof of
    Theorem~\ref{theorem-uniform-SWSTS-in-BPP}.
\end{proof}

It is straightforward to show that any protocol in $\POP$ satisfies the necessary
symmetry properties, and that its transition relation can be computed with
a probabilistic log-space transducer. We can also easily test $s≤t$
for fixed $s$ in logarithmic space. So any language accepted by $\POP$
with a symmetric log-space-computable input map is a symmetric language in
$\classBPL$.

Assuming appropriate input maps, similar results can be shown for
closed $\CRN$ protocols with rational reaction rates and for
$\GOSSIP$, $\MATCHING$, and $\SHUFFLE$. For the synchronous gossip
models, the sampling is a bit more complicated than for $\POP$ or
$\CRN$, because we must be careful not to use any agent or token more
than once. But this can be done by iterating appropriate sampling
without replacement using counters of unused agents in the initial
configuration that are decrement upon use, while accumulating counters
of agents seeing particular values or combinations of tokens that are
incremented as needed with each iteration. This only works because the
set of possible views and previous states for an agent at each step in
one of these models has constant size.

For $\PUSH$ we must be more careful. The obvious scheme of tracking
which states are pushed to which agents fails because each agent may
see anywhere from $0$ to $n$ states, potentially giving more views
than can be stored in logarithmic space. However, if we assume that
the transition function for a single agent depends only on the counts
of seen states, we can sample the view of each target agent $a$
separately by sampling for each previously unseen agent whether it
pushes to $a$ or not; agents that push to $a$ are included in the
appropriate count, while agents that don't are thown back in the pool
to be sampled by the next agent. For a symmetric, log-space computable
transition function at each agent, this satisfies the requirements of
the theorem, showing that this particular constrained version of
$\PUSH$i accepts only symmetric languages in $\classBPL$.

\section{Simulating Turing machines with SWSTS protocols}
\label{section-simulating-TM}

Population protocols~\cite{AngluinAE2008fast} and chemical reaction
networks~\cite{CookSWB2009} are known to be
able to simulate probabilistic Turing machines, often via intermediate
representations like register machines.
In this section, we show that
many SWSTS protocols that either have or can encode a linear structure on
agents can simulate Turing machines directly with small probability of
error at each step, giving them the ability to compute any predicate
in $\classBPP$.

We will concentrate first on $\POP$ and its
edge-oracle variants, then reduce to these cases by showing how to
simulate $\POP^A$ in various other SWSTS models with the same
edge oracle, or with some other structure equivalent in power to an
appropriate edge oracle as in $\BETWEEN$.

\subsection{Simulating a Turing machine in $\POP^A$}
\label{section-TM-in-POP}

For the basic population protocol model, established
results~\cite{AngluinADFP2006,AngluinAE2008fast} show how to implement
a register machine that is almost equivalent in power to a probabilistic
Turing machine with a log-sized work tape, with a polynomially-small
probability of error at each step. The ``almost'' here is a result of
not being able to represent the order of the symbols in the input string:
because of symmetry between agents, these simulated machines can only observe
their inputs up to an arbitrary permutation. Nonetheless, these
results show that $\POP$ can compute any symmetric language in $\BPL$.
Combined with Theorem~\ref{theorem-log-space-uniform-SWSTS-in-BPL},
this shows that unmodified $\POP$ computes precisely the
symmetric languages in $\classBPL$.

Adding an edge oracle may significantly increase the power of the
model. In $\POP^<$, most agents have an immediate predecessor
and successor, and we can use this structure to simulate a Turing
machine tape with each cell assigned to some agent. Because the agents
cannot observe the successor relation directly (which would produce a
non-well-structured transition system, as discussed in
§\ref{section-successor}), the agents will
instead need to repeatedly run a subprotocol that finds the predecessor or
successor of the current tape cell. This is not difficult,
but it requires waiting for the subprotocol to finish, which 
in turn requires building some sort of clock mechanism and allowing
for a polynomial probability of failure if the clock ticks before the
subprotocol converges.
An additional issue is that we will need both a designated leader and
a number of agents polynomial in the size of the input to represent
the entire tape, but these can be supplied by the input map.
Details of the resulting construction are given in
Section~\ref{section-POP-ordered-BPP}, where it is shown that
$\classBPP ⊆ \POP^{<}$.

As it happens, $\POP^{=}$ also has enough power to simulate a
probabilistic polynomial-time Turing machine.
Because the edge oracle does not provide an order on agents, we can
instead represent the $i$-th tape cell as a collection of $i$ agents
in the same equivalence class, and use comparison protocols to move
between $i$ and $i-1$ or $i+1$ as needed. This requires a bit more
setup in the input map to supply the $Θ(m^2)$ agents needed to encode
$m$ tape cells and to set up their colors to group them appropriately.
The full construction is given in
Section~\ref{section-POP-unordered-BPP}, where it is shown that
$\classBPP ⊆ \POP^=$.

\subsubsection{Simulating a Turing machine in $\POP^<$}
\label{section-POP-ordered-BPP}

Given a probabilistic Turing machine that finishes in $n^c$ steps, we
will use the input map to construct a simulated tape
with $n^c$ agents, plus a handful of additional utility agents. These
additional agents include a leader agent $L$ that coordinates the
simulation and holds the state of the Turing machine head, a timer
agent $T$ used to implement a simple clock as
in~\cite{AngluinADFP2006}, and two agents $R_0$ and $R_1$ used to
implement a simple coin-flipping mechanism.
It is easy to see that this input map is both order-preserving and
log-space-computable.

Most of the work is done by the leader agent, which will be directly
responsible for simulating the TM head and updating the contents of
tape agents. The tape agents will nonetheless have to carry out a
protocol for finding the predecessor or successor of the current head
position.

This protocol is described in
Table~\ref{table-predecessor-successor}. For simplicity, the
table only includes transitions where the initiator
is in a position less than that of the responder, but we imagine that
for each such transition there is a symmetric transition where the
positions of the agents are reversed. We also omit transitions that
have no effect and any transitions involving the leader agent, whose
role is limited to switching the state of the central agent after an
appropriate amount of time.

\begin{table}
    \begin{displaymath}
        \begin{array}{r@{\,<\,}l@{\,\,→\,\,}r@{\,<\,}l@{\quad}l}
            b&0&-&0& \text{generate $-$} \\
            -&0&-&0& \\
            +&0&-&0& \\
            \hline
            0&b&0&+& \text{generate $+$}\\
            0&-&0&+& \\
            0&+&0&+& \\
            \hline
            0&-&0&b& \text{remove erroneous $-$}\\
            +&0&b&0& \text{remove erroneous $+$}\\
            \hline
            -&-&b&-& \text{prune $-$}\\
            +&+&+&b& \text{prune $+$}\\
        \end{array}
    \end{displaymath}
    \caption{Finding predecessors and successors in $\POP^{<}$}
    \label{table-predecessor-successor}
\end{table}

Initially, the agent at the current head position is assigned state
$0$. This causes it to recruit agents to its left and right in states
$-$ and $+$ respectively. After waiting long enough that the immediate
predecessors and successors are likely to have been recruited, the
central agent is switched to the blank state $b$. Excess $-$ and $+$
agents are then pruned until only the immediate predecessor and
successor remain.

Because each step during the first phase recruits the immediate
predecessor with probability $Θ(N^{-2})$, waiting $Ω(N^3)$ steps gives
an exponentially low probability that the immediate predecessor is not
recruited.\footnote{There is a possibility, if agents are marked with
the wrong direction,
that a newly recruited predecessor or successor will be pruned by an
erroneously marked agent that has not yet been removed. But standard
coupon collector bounds give an exponentially small probability that
any erroneous agents remain after $\widetilde{O}(N^2)$ steps, and
it happens that in our TM simulation the issue does not arise because
no such erroneous off-side agents are ever created.} 
The same holds for the immediate successor. For the second
phase, we can use the analysis from~\cite{AngluinAE2008fast} of leader election by
cancellation to argue that the expected time for the number of $-$ or
$+$ agents to reduce to one is $O(N^2)$. Again waiting $Ω(N^3)$ steps
gives an exponentially small probability of failure.

Unfortunately, Lemma~\ref{lemma-polynomial-path} tells us that any
attempt to wait for more than a constant number of steps in any SWSTS
will fail with at least polynomial probability. So to wait $Ω(N^3)$
steps, we cannot do much better than the simple clock
of~\cite{AngluinADFP2006}, where the leader agent waits to encounter a
unique timer agent $k$ times in a row without seeing any other agents.
As observed in~\cite[Theorem 24]{AngluinADFP2006}, this takes
$O(N^{k+1})$ steps on average; a simple union bound also shows that
the probability that it finishes in $O(N^\ell)$ steps is $O(N^{\ell-k-1})$.
So for appropriate choice of $k$, the $Ω(N^3)$-step delays needed in
Table~\ref{table-predecessor-successor} can be carried out with any
desired polynomial probability of failure.

We can now describe the Turing machine simulation. For each step, the
leader first flips a coin by waiting to see either $R_0$ or $R_1$. It
then waits to encounter the tape agent marked as holding the current
head position, updates the tape symbol stored there according to the
simulated TM's transition function, and replaces the head mark with a
$0$ mark for the neighbor-finding protocol described in
Table~\ref{table-predecessor-successor}. It then
carries out this protocol to find the
immediate predecessor and successor of the previous head position.
Depending on whether the head is moving left or right, the leader
waits to encounter either the $-$ agent or $+$ agent and marks it as
the new head position.

We can thus simulate one step of a probabilistic Turing machine with
any polynomially small failure probability. By adjusting the exponent
appropriately and taking a union bound, we get polynomially
small failure probability even after $n^c$ steps. This shows that
$\classBPP ⊆ \POP^{<}$.

Inclusion in the other direction is immediate from
Theorem~\ref{theorem-uniform-SWSTS-in-BPP}. So we have:
\begin{theorem}
    \label{theorem-BPP-ordered-POP}
    $\POP^{<} = \classBPP$.
\end{theorem}

\subsubsection{Simulating a Turing machine in $\POP^=$}
\label{section-POP-unordered-BPP}

For this construction we will be particularly profligate in our use of
agents. Recall that $\POP$ with $n$ agents is capable of simulating,
with polynomial error per step, a probabilistic Turing machine with
$O(\log n)$ space~\cite{AngluinAE2008fast}. To implement a Turing
machine that finishes in $n^c$ steps in $\POP^{=}$, we will create
$n^c$ equivalence classes of $Θ(n^c)$ agents each, where each
equivalence class will implement a separate Turing machine capable of
storing the numerical index of a particular tape cell in addition to its
contents and a possible head state. As in the $\POP^{<}$ construction,
setting up these Turing machines and initializing the indices and
states can be done using an order-preserving log-space-computable input
map.

To simulate the full TM, whichever equivalence class $i$ of agents
currently holds the head computes the next step and updates the local
tape value appropriately. It must then find either its predecessor or
successor to hand off control of the head. This can be done by (a)
randomly choosing some other equivalence class as the candidate new
position, then (b) testing if the index of this other equivalence is
$i-1$ or $i+1$ as needed. If this test fails, try again.

Randomly choosing the candidate equivalence class is simply a matter
of having the leader in $i$ wait to encounter a leader in some other
equivalence class, which takes polynomially many expected
interactions. After marking this candidate leader for future
reference, testing if the other class has index $i±1$ can then be
done using the techniques of~\cite{AngluinAE2008fast} in polynomially
many expected interactions by having the leaders coordinate
decrementing a copy of their stored indices until one hits zero and
testing the value of the other copy. This incurs a substantial
overhead since the leaders must wait to communicate directly between each
decrement, but it's still polynomial. If the test fails, the candidate
leader is unmarked and a new candidate is chosen. This will find the
correct predecessor or successor class after a polynomial expected
number of attempts.

Error analysis is similar to the $\POP^{<}$ case: we just need to
set the polynomial probability of error at each step small enough
so that it is still polynomially small even when multiplied by the
large but still polynomial number of steps. This gives $\classBPP ⊆
\POP^{=}$, and using Theorem~\ref{theorem-uniform-SWSTS-in-BPP}
for the other direction gives
\begin{theorem}
    \label{theorem-BPP-unordered-POP}
    $\POP^{=} = \classBPP$.
\end{theorem}

\subsection{Simulating $\POP^A$ in other models}

To show that other SWSTS models can compute functions in $\classBPL$
or $\classBPP$, we simulate $\POP^A$. Typically this involves showing
that a model can carry out a loop where it repeatedly chooses two
agents to act as the initiator and responder in a $\POP$ interaction,
and uses whatever oracle is available to the model to simulate the
$\POP^A$ oracle.

How difficult this is will depend somewhat on how mismatched the
models are from $\POP$. The two main difficulties we may need to
overcome are that a model may not include the direct interactions
between agents that allows $\POP$ to update both agents' states
simultaneously, and that in models with indirect observation (like
$\GOSSIP$), more than one agent may observe a single agent. We start
(in §\ref{section-pop-in-gossip}) by showing how to overcome both
these issues when simulating $\POP^A$ in $\GOSSIP^A$. Simulations
by other models 
(described in subsequent sections) will allow simpler constructions.

\subsubsection{Simulating $\POP^A$ in $\GOSSIP^A$}
\label{section-pop-in-gossip}

Simulating one step of $\POP^A$ in $\GOSSIP^A$ will involve a sequence
of phases coordinated by a unique leader agent, which uses a unique
timer agent to implement a variant of the timer mechanism
of~\cite{AngluinADFP2006}. The remaining $n$ follower agents represent
agents in the $\POP^A$ protocol. Setup of the initial states of these
$n+2$ agents is handled by the input map.

To avoid confusion between the $\GOSSIP$ protocol and the protocol it
is simulating, we will describe the component of each $\GOSSIP$
agent's state that controls its behavior as its role.

Because $\GOSSIP$ allows many agents to pull data from a single agent,
it is not trivial to pick out a single initiator and responder in each
phase. Instead, the leader may recruit many candidates for these
roles, then apply a leader-election-like algorithm to remove
duplicates.

An additional complication is that the chosen initiator and responder
agents do not interact directly; instead, the responder will wait to
sample the initiator and calculate the new state of both agents, using
the gossip mechanism to communicate this change to the initiator. The
choice of having the responder do the computation is required by the
orientation of the edge oracle: under the convention that the oracle
reports $A(x,y)$ where $x$ is the transmitter of information, it is
only the responder that has access to the correct value of the oracle.

For the leader agent, the role will cycle through the values
\DS{recruit-initiator},
\DS{recruit-responder},
\DS{prune-extras},
\DS{update},
and
\DS{reset},
with the first roles being maintained for one step and the others for
an expected $Θ(n^k)$ steps 
by having the leader wait to see the
unique timer agent $k$ times in a row. An additional state component
in the leader is used to maintain the count of the current streak of
observations of the timer agent.

For follower agents, each agent maintains three components:
\begin{itemize}
    \item \DS{role} is its current role in the simulating
        protocol.
    \item $q$ is its state in the simulated $\POP$ protocol.
    \item $q'$ is an extra state component used only by a simulated
        initiator to communicate an updated state to a simulated
        responder.
\end{itemize}

In the initial configuration, each follower is in the $\DS{idle}$
role and has state $q$ equal to the state in $\POP^A$ chosen by the
input map.

Table~\ref{table-gossip-transitions} shows how these components
evolve in follower agents in response to observations of the leader
and each other.

\begin{table}
    \centering
    \begin{tabular}{lll}
        Role & Observed state & Update rule
        \\
        \hline
        \DS{idle} & \DS{recruit-initiator} &
        $\DS{role} ← \DS{initiator}$
        \\
        \hline
        \DS{idle} & \DS{recruit-responder} &
        $\DS{role} ← \DS{candidate-responder}$
        \\
        \DS{initiator} & \DS{recruit-responder} &
        $\DS{role} ← \DS{candidate-responder}$
        \\
        \hline
        \DS{initiator} & \DS{initiator}
        & $\DS{role} ← \DS{idle}$
        \\
        \DS{candidate-responder} & \DS{candidate-responder}
        & $\DS{role} ← \DS{idle}$
        \\
        \hline
        \DS{candidate-responder} & \DS{update}
        & $\DS{role} ← \DS{responder}$
        \\
        \DS{responder} & \DS{initiator}, $q_i$, $A$
        & $\DS{role} ← \DS{sending}$; $(q',q) ← δ(q_i,q,A)$
        \\
        \DS{initiator} & \DS{sending}, $q'$
        & $\DS{role} ← \DS{idle}$; $q ← q'$
        \\
        \hline
        \DS{initiator} & \DS{reset}
        & $\DS{role} ← \DS{idle}$
        \\
        \DS{candidate-responder} & \DS{reset}
        & $\DS{role} ← \DS{idle}$
        \\
        \DS{responder} & \DS{reset}
        & $\DS{role} ← \DS{idle}$
        \\
        \DS{sending} & \DS{reset}
        & $\DS{role} ← \DS{idle}$
        \\
    \end{tabular}
    \caption{Follower transitions for simulating $\POP^A$ in $\GOSSIP^A$}
    \label{table-gossip-transitions}
\end{table}

\begin{lemma}
    \label{lemma-gossip-pop}
    Fix $c ≥ 2$ and let $k=3c$.
    Then the protocol described above simulates $O(n^c)$ steps of a
    $\POP^A$ protocol with transition function $δ$ in expected time
    $O(n^{4c})$ with a probability
    of failure of $O(n^{-c})$.
\end{lemma}
\begin{proof}
    Let's look at what happens within one cycle consisting of five
    phases corresponding to the leader's roles.

    Each cycle can fail in one of two ways. A \concept{global failure}
    irrecoverably breaks the simulation: if a global failure occurs,
    the protocol is over and the simulation of $\POP^A$ has failed. We
    will show that global failures occur with probability only
    $O(n^{-2c})$ per cycle, or $O(n^{-c})$ across all $Θ(n^c)$
    simulated steps. A \concept{local failure} is less
    destructive and only prevents the current cycle from making
    progress. These may occur with constant probability per cycle and
    have the effect of slowing the simulation down by a constant
    factor on average.

    The phases in each cycle proceed as follows:
    \begin{enumerate}
        \item 
            During the single step where the leader has role
            \DS{recruit-initiator}, each of the
            $n$ follower agents has an independent $\frac{1}{n+2}$
            probability of
            observing this role and switching to
            \DS{initiator}.
            This gives a probability of
            $1-\parens*{1-\frac{1}{n+2}}^{n}$ that at least one agent
            switches to \DS{initiator}, which converges
            to the constant $1-1/e$ in the limit.

            This phase does not produce global failures but may
            produce a local failure if it ends in a configuration with
            no agents in role \DS{initiator}.
        \item 
            During the following step where the leader has role
            \DS{recruit-responder}, the same argument shows a
            constant probability of getting at least one agent in role
            \DS{candidate-responder}.

            The rule allowing $\DS{initiator}$ to be overwritten
            with $\DS{candidate-responder}$ exists only to make
            this argument possible. However, it has the unfortunate
            side effect of allowing all initiators from the previous
            phase to be overwritten. This event can only occur at each
            agent with probability $\frac{1}{(n+2)^2}$, since it
            requires the agent to observe the leader in both steps.
            The union bound then gives a probability of all initiators
            being overwritten of at most $\frac{n}{(n+2)^2} = O(1/n)$.
            Overwriting all initiators results in a local failure.

            This phase also cannot produce a global failure but it can,
            with at most constant probability, end
            in a state with no agents in role \DS{initiator} or
            no agents in role \DS{candidate-responder}. Either of
            these outcomes will be a local failure that causes the
            current cycle to have no effect.
        \item
            During the next $Θ(n^k)=Θ(n^{3c})$ expected steps where the leader
            has role \DS{prune-extras},
            as long as there is more than one
            \DS{candidate-responder}
            agent, each such agent has at least a $\frac{1}{n+2}$ chance
            per step of observing another agent in the same role and
            switching back to \DS{idle}.
            So for each candidate, after an expected $O(n)$ steps it
            either drops out or is the only remaining candidate.
            However, there are still several ways this
            process can go wrong:
            \begin{enumerate}
                \item
                    The actual time to remove excess candidates
                    exceeds $O(n^c)$. 
                    By Markov's inequality, each excess candidate has
                    at least a constant probability of dropping out
                    within each interval of $O(n)$ steps, so after
                    $O(n^2) = O(n^c)$ steps the
                    chance that any excess candidate survives is
                    exponentially small.

                    A failure of this kind will be treated as a global
                    failure.
                \item 
                    The time before reaching the
                    \DS{update} phase (and thus
                    switching a \DS{candidate-responder} agent to
                    \DS{responder}) is too short. Here we can use
                    a union bound to argue that the probability that
                    this occurs is $O(n^{c}⋅n^{-k}) = O(n^{-2c})$.

                    A failure of this kind will also be treated as
                    a global failure.
                \item 
                    The last remaining candidates observe each other
                    in the same step and all drop out together.
                    Let $X_t$ be the number of surviving candidates
                    after $t$ steps and let $τ$ be last step with $X_τ
                    > 1$.
                    Then
                    \begin{align*}
                        \Prob{X_{τ+1} = 0}
                        &=
                        ∑_{\ell=2}^{n}\Prob{X_τ=\ell}⋅\ProbCond{X_{τ+1} = 0}{X_τ=\ell}
                        \\
                        &=
                        ∑_{\ell=2}^{n}\Prob{X_τ=\ell}⋅
                        \parens*{
                            \frac{
                                \parens*{\frac{\ell-1}{n+2}}^\ell
                            }{
                                \parens*{\frac{\ell-1}{n+2}}^\ell + \ell
                                \parens*{\frac{\ell-1}{n+2}}^{\ell-1}
                                \parens*{1-\frac{\ell-1}{n+2}}
                            }
                        }
                        \\
                        &≤
                        ∑_{\ell=2}^{n}\Prob{X_τ=\ell}⋅\frac{1}{2}
                        \\
                        &=
                        \frac{1}{2}.
                    \end{align*}
                    The same analysis applies to agents in role
                    \DS{initiator}. Assuming that we start
                    the \DS{prune-extras} phase with at least one
                    agent in each of the roles
                    \DS{candidate-responder} and
                    \DS{initiator}, we get a probability
                    of $O(n^{-2c})$ that more than one agent remains
                    in each of these roles at the end of this phase,
                    and at least $1/2$ that at least one agent remains
                    in each role, giving a probability of
                    at least
                    $1/4$ that at least one agent remains in each
                    role.

                    This gives a global failure with probability
                    $O(n^{-2c})$ and a local failure with probability
                    at most $3/4$ conditioned on no prior failures in
                    this cycle.
            \end{enumerate}
        \item 
            During the next $Θ(n^k)=Θ(n^{-3c})$ expected steps, the leader
            has role \DS{update}. Let's suppose that previous
            phases have worked as intended, so we have exactly one
            agent $i$ in role $\DS{initiator}$
            and one agent $r$ in role $\DS{candidate-responder}$.

            Note that by symmetry, each ordered pair of follower
            agents is equally likely to be $(i,r)$. So the simulated
            process matches the uniform scheduling assumption of the
            $\POP^A$ model.

            We can consider three sequential sub-phases:
            \begin{enumerate}
                \item 
                    After an expected $O(n)$ steps, $r$ will observe
                    $\DS{update}$ and switch to
                    $\DS{responder}$.
                \item 
                    After an expected $O(n)$ additional steps, $r$ will
                    observe $i$, compute the output of the joint transition
                    function $δ(q_i, q_r, A(i,r))$, and update its own state
                    accordingly. The extra component $q'$ will store the
                    updated state for $i$, and the role of $r$ will
                    switch to $\DS{sending}$ to let $i$ know that
                    this value is current.
                \item 
                    After a further $O(n)$ expected steps, $i$ will observe
                    $r$ in role $\DS{sending}$, update its state to $q'$, and switch to role
                    $\DS{idle}$. At this point both $i$ and $r$ have the
                    same state they would have chosen in $\POP^A$.
            \end{enumerate}

            These three intervals sum to $O(n)$ expected steps. By the same
            argument as for the previous phase, the probability that
            the phase finishes before these intervals are complete is
            $O(n^{-2c})$.

            The case where there is more than one agent in role
            \DS{candidate-responder} or more than one agent in
            role \DS{initiator} will already have been counted as
            a global failure, so we do not need to worry
            about it again.

            This leaves the case where either $i$ or $r$ is missing.
            If only $i$ is missing, $r$ will never update its state in
            response to seeing $\DS{initiator}$ and so will end
            this subphase in either $\DS{candidate-responder}$
            or $\DS{responder}$; in either case it will not
            update its state.
            If only $r$ is missing, $i$ will never see an agent in role
            $\DS{sending}$ and thus will also perform no
            updates. If both are missing, nothing happens. All three
            cases are local failures.
        \item 
            During the final $Θ(n^k)$ expected steps where the leader
            has role \DS{reset}, any remaining agents in roles
            $\DS{candidate-responder}$, $\DS{responder}$,
            $\DS{sending}$, or
            $\DS{initiator}$ will each reset to $\DS{idle}$
            within $O(n)$ expected steps. In a successful phase this
            will only be one agent in role $\DS{sending}$, but the
            rules for this phase will also clean up any leftover
            agents in the other roles if a previous phase produced a
            local failure.

            Assuming that no global failure previously occurred,
            there is at most one agent in each of the target
            roles, and the probability of failure of this phase is $O(n^{-2c})$.
    \end{enumerate}

    We can sum up all the global-failure probabilities with a union
    bound to get an $O(n^{-2c})$ probability per cycle of a global
    failure. For local failures, each phase or subphase avoids a
    local failure with at least constant probability conditioned on
    no previous failures; multiplying out these conditional
    probabilities gives a constant probability that there is no local
    failure anywhere in the cycle.

    Local failures are just missed steps, so they only affect the constant in
    the rate of the simulation. It follows that simulating $Θ(n^{c})$ steps of
    $\POP^A$ requires only $Θ(n^c)$ cycles of the $\GOSSIP^A$
    protocol, which corresponds to $O(n^{c+k}) = O(n^{4c})$ steps.
    The total probability of a global failure over these $Θ(n^c)$ cycles is
    $O(n^c)⋅O(n^{-2c}) = O(n^{-c})$.
\end{proof}

It is not hard to see that each of $\GOSSIP^<$ and $\GOSSIP^=$
satisfies the uniformity conditions of
Theorem~\ref{theorem-uniform-SWSTS-in-BPP}, which puts both of these
classes in $\classBPP$. We can also combine
Lemma~\ref{lemma-gossip-pop} with the simulations in
Theorems~\ref{theorem-BPP-ordered-POP}
and~\ref{theorem-BPP-unordered-POP} to show the reverse inclusion.
This gives:
\begin{theorem}
    \label{theorem-BPP-GOSSIP}
    $\GOSSIP^< = \GOSSIP^= = \classBPP$.
\end{theorem}

With no oracle, Lemma~\ref{lemma-gossip-pop} shows that $\POP ⊆
\GOSSIP$. Combined with
Theorem~\ref{theorem-log-space-uniform-SWSTS-in-BPL}, this shows
that $\GOSSIP$ computes the symmetric languages in $\classBPL$.

\subsubsection{Simulating $\POP^A$ in $\PUSH^A$ or $\SHUFFLE^A$}

In $\PUSH^A$, each agent is only visible to one other agent in each
step. This allows for a simple simulation of $\POP^A$ based on passing
a single active token around the agents. We will still
need a leader agent, but its role is limited to being a neutral
arbiter that selects an initiator uniformly in each cycle, avoiding
bias that might otherwise arise if the initiator were selected by an
agent chosen for some role in a previous cycle. Essentially the same
algorithm works in $\SHUFFLE^A$ by treating the simulated token as an
explicit token, but here we also have to handle the case where an
agent receives the token it sent.

The transition table for this simulation is given in
Table~\ref{table-pop-with-token-passing}. Each group of transitions in the table
implements movement of the active token, which is represented by a
role $\DS{start}$, $\DS{initiator}$, $\DS{wait+restart}$,
$\DS{restart}$, $\DS{send}$, $\DS{wait+send}$,
or $\DS{done}$.
The additional roles $\DS{wait}$ and $\DS{receive}$ mark agents
that have a special role that affects how they process an incoming
token. Each agent also stores its simulated state $q$ and an optional extra
state $q'$ to be transmitted to the simulated initiator for this
cycle.

A transition is controlled by the current role of the agent and the
active token (if any) that it observes. In the $\PUSH^A$ model, this
observation consists of receiving a push. In the $\SHUFFLE^A$ model,
this observation consists of receiving a token. In the latter case we
assume that only the agent with the active token emits a non-null
token, while the other agents emit default null tokens. In the
$\PUSH^A$ model the agents cannot choose whether to emit a push or
not, but we assume that pushes of roles that do not correspond to the
active token are ignored.

We omit from the table transitions where an agent observes the same
role that it already has; in this case, the agent does not update its
role and tries again in the next step. These transitions are only an
issue in the $\SHUFFLE^A$ model.

\begin{table}
    \begin{displaymath}
        \begin{array}{cccc}
            \text{Current role} & \text{Observed} & \text{Update rule} & \text{New role}\\
            \hline
            \DS{start} & & & \DS{wait} \\
            \DS{idle} & \DS{start} & & \DS{initiator}(q) \\
            \hline
            \DS{initiator} & & & \DS{receive}
            \\
            \DS{wait} & \DS{initiator}(q_i) & & \DS{wait+restart} \\
            \DS{idle} & \DS{initiator}(q_i),A & 
            (q,q') ← δ(q_i,q,A) & \DS{send}(q')
            \\
            \hline
            \DS{wait+restart} & & & \DS{wait} \\
            \DS{idle} & \DS{wait+restart} & & \DS{restart} \\
            \DS{receive} & \DS{wait+restart} & & \DS{initiator}(q) \\
            \hline
            \DS{restart} & & & \DS{idle} \\
            \DS{idle} & \DS{restart} & & \DS{restart} \\
            \DS{wait} & \DS{restart} & & \DS{wait+restart} \\
            \DS{receive} & \DS{restart} & & \DS{initiator}(q) \\
            \hline
            \DS{send}(q') & & & \DS{idle}
            \\
            \DS{idle} & \DS{send}(q') & & \DS{send}(q')
            \\
            \DS{wait} & \DS{send}(q') & & \DS{wait+send}(q')
            \\
            \DS{receive} & \DS{send}(q') & q ← q' & \DS{done} \\
            \hline
            \DS{wait+send}(q') & & & \DS{wait}
            \\
            \DS{idle} & \DS{wait+send}(q') & & \DS{send}(q')
            \\
            \DS{receive} & \DS{wait+send}(q') & q ← q' & \DS{done} \\
            \hline
            \DS{done} & & & \DS{idle} \\
            \DS{wait} & \DS{done} & & \DS{start} \\
        \end{array}
    \end{displaymath}
    \caption{Simulating $\POP^A$ by token passing}
    \label{table-pop-with-token-passing}
\end{table}

In the initial configuration, the leader agent has role \DS{start} and
the follower agents have role \DS{idle}, with each follower agent's
$q$ field set to the state of the $\POP^A$ agent it is simulating.

A normal cycle of the protocol simulates one step of $\POP^A$ and
proceeds as follows:
\begin{enumerate}
    \item The leader transmits its $\DS{start}$ token to a follower
        that adopts the role $\DS{initiator}$, making it the initiator
        of the simulated step. The leader switches to the $\DS{wait}$
        role until it is notified that the cycle is complete.
    \item The initiator transmits its $\DS{initiator}$ token to a
        simulated responder, which uses the accompanying state and the
        output of the oracle $A$ to compute the transition relation
        $(q',q) ← δ(q_i,q,A)$ for the simulated population protocol.
        The responder enters a $\DS{send}(q')$ role to transmit the
        updated initiator state back to the initiator, while the
        initiator switches to the $\DS{receive}$ role to await
        receiving this state.
    \item The $\DS{send}(q')$ token wanders through the population
        (possibly switching temporarily to a $\DS{wait+send}(q')$
        token if held by the leader) until it reaches the initiator.
        At this point the initiator updates its state and emits a
        $\DS{done}$ token.
    \item The $\DS{done}$ token similarly wanders through the
        population until it reaches the leader and starts the next
        cycle by switching the leader's role to $\DS{start}$.
\end{enumerate}

It is not hard to see that a normal cycle of the protocol chooses each
possible initiator and responder with equal probability
$\frac{1}{n(n+1)}$ and that it correctly updates the states of both
agents before starting the next cycle.

A complication is that if the $\DS{initiator}$ token passes back to
the leader instead of to another agent, the leader cannot act as a
responder, so it needs to restart the initiator-responder interaction.
It does so by switching to a $\DS{wait+restart}$ role that propagates
through the population as either $\DS{wait+restart}$ or $\DS{restart}$
until it reaches the previously selected initiator, which switches
back from $\DS{receive}$ to $\DS{initiator}$.

The classes of configurations of the protocol and the transitions
between them are shown in
Figure~\ref{fig-state-machine-for-pop-with-tokens}. Each class is
labeled by the role of the agent with the active token first, followed
by any other non-\DS{idle} roles that are present in the population.
Each transition in the chain occurs with probability at least $\frac{1}{n+1}$,
so the $\DS{start}$ configurations are revisited
infinitely often with probability $1$, which implies that the
simulated $\POP^A$ protocol takes infinitely many steps. With a bit of
work it is also possible to show that the expected slowdown is $O(n)$,
mostly incurred while waiting for the $\DS{send}$ or $\DS{done}$
tokens to reach the right agent.
\begin{figure}
    \centering
    \begin{tikzpicture}
        \node[draw,rectangle] (start) at (-0.5,1) { \DS{start} };
        \node[draw,rectangle] (initiator) at (2,0) { 
        \begin{tabular}{c}
            \DS{initiator}$(q)$ \\
            \DS{wait}
        \end{tabular}
        };
        \node[draw,rectangle] (send) at (5,0.5) { 
        \begin{tabular}{c}
            \DS{send}$(q')$ \\
            \DS{wait} \\
            \DS{receive}
        \end{tabular}
        };
        \node[draw,rectangle] (ws) at (5,2.5) { 
        \begin{tabular}{c}
            \DS{wait+send}$(q')$ \\
            \DS{receive}
        \end{tabular}
        };
        \node[draw,rectangle] (done) at (2,2) { 
        \begin{tabular}{c}
            \DS{done} \\
            \DS{wait}
        \end{tabular}
        };
        \node[draw,rectangle] (wr) at (0.5,-2.5) { 
        \begin{tabular}{c}
            \DS{wait+restart} \\
            \DS{receive}
        \end{tabular}
        };
        \node[draw,rectangle] (restart) at (3.5,-2.5) { 
        \begin{tabular}{c}
            \DS{restart} \\
            \DS{wait} \\
            \DS{receive}
        \end{tabular}
        };
        \draw[->,thick,blue] (start) edge (initiator);
        \draw[->,thick,blue] (initiator) edge (send);
        \draw[->,thick,blue] (send) edge (done);
        \draw[->,thick,blue] (done) -- (start);
        \draw[<->] (initiator) -- (wr);
        \draw[<->] (wr) -- (restart);
        \draw[->] (restart) -- (initiator);
        \draw[<->] (ws) -- (send);
        \draw[->] (ws) -- (done);
    \end{tikzpicture}
    \caption{Transition diagram of simulation of $\POP^A$ in $\PUSH^A$ or
    $\SHUFFLE^A$. Typical execution cycle is highlighted in blue.}
    \label{fig-state-machine-for-pop-with-tokens}
\end{figure}
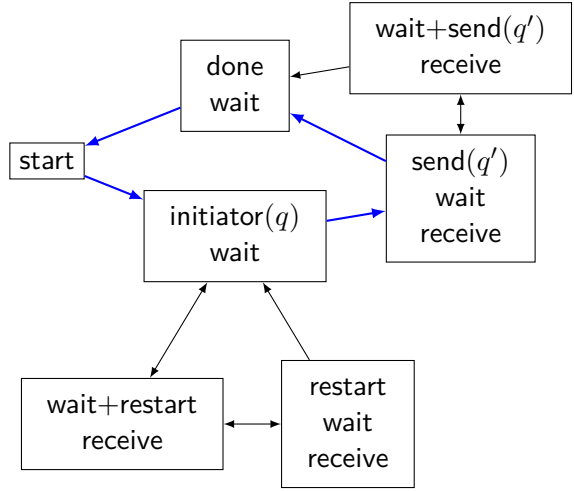

Unlike the $\GOSSIP^A$ simulation, this simulation experiences no
unrecoverable errors,
and can continue to simulate $\POP^A$ in $\PUSH^A$ or $\SHUFFLE^A$
indefinitely. We can express this as:
\begin{lemma}
    \label{lemma-pop-in-push-or-shuffle}
    For any edge oracle $A$,
    $\POP^A ⊆ \PUSH^A$ and $\POP^A ⊆ \SHUFFLE^A$.
\end{lemma}

It follows that
\begin{theorem}
    \label{theorem-BPP-PUSH}
    $\PUSH^< = \PUSH^= = \SHUFFLE^< = \SHUFFLE^= = \classBPP$,
    assuming $\PUSH^<$ and $\PUSH^=$ are restricted to 
    polynomial-time-computable transition functions.
\end{theorem}

As with $\GOSSIP$, a similar reduction to $\POP$ shows that with no
oracle and appropriate restrictions, $\PUSH$ and $\SHUFFLE$ accept
precisely the symmetric languages in $\classBPL$.

\subsubsection{Simulating $\POP^A$ in $\MATCHING^A$}

It is possible to simulate $\POP^A$ in $\MATCHING^A$ using the same
token-passing approach as in $\PUSH^A$ or $\SHUFFLE^A$, but there is a
simpler simulation that takes advantage of the direct interactions
between agents.

A transition rule for this simulation is given in
Table~\ref{table-MATCHING-POP}.
The idea is to use a single leader agent to randomly mark an
initiator and responder for each simulated step. These agents then
interact directly, marking themselves as done. The leader collects
these marks before
choosing a new initiator and responder.
Because the protocol requires no timing assumptions, no separate timer agent is
needed.

\begin{table}
    \begin{align*}
        \ell_0, \DS{idle}(q) &→ \ell_1, \DS{initiator}(q) \\
        \ell_1, \DS{idle}(q) &→ \ell_2, \DS{responder}(q) \\
        \DS{initiator}(q_1), \DS{responder}(q_2), A &→
        \DS{done}(q'_1), \DS{done}(q'_2) \quad \text{where
        $q' = δ(q_1,q_2,A)$} \\
        \ell_2, \DS{done}(q) &→ \ell_3, \DS{idle}(q) \\
        \ell_3, \DS{done}(q) &→ \ell_0, \DS{idle}(q)
    \end{align*}
    \caption{Transition rules for simulating $\POP^A$ in $\MATCHING^A$}
    \label{table-MATCHING-POP}
\end{table}

It is straightforward to demonstrate that this simulates a protocol in
$\POP^A$ with no error and $O(n^2)$ slowdown.

Combining this simulation with our earlier results gives:
\begin{theorem}
    \label{theorem-BPP-MATCHING}
    $\MATCHING^< = \MATCHING^= = \classBPP$.
\end{theorem}

We also have that $\MATCHING$ accepts the symmetric languages in
$\classBPL$, assuming a log-space computable input map.

\subsubsection{Simulating $\POP^<$ in $\BETWEEN$}

Simulating $\POP^<$ in $\BETWEEN$ is straightforward, because we use
the input map to choose a graph in the form of a path, where each
agent's position on the path corresponds to the total order $<$.

Given a $\POP^<$ protocol using $n$ agents with transition function
$δ$, construct a $\BETWEEN$ protocol with $n+2$ agents by extending
the input map to add new
agents $L$ and $R$ and the left and right ends of the $\POP^<$
configurations, and organizing these nodes into a path according to
the total order $<$.

Recall that any
transition in $\BETWEEN$ involves a triple of agents $\Tuple{s,t,u}$
where $t$ is on a simple path from $s$ to $u$. Let a transition
involving any such triple
that does not include exactly one of $L$ or $R$ have no effect. For a
triple $\Tuple{L,a,b}$, update the states of $a$ and $b$ to
$δ(a,b,<)$. For a triple $\Tuple{a,b,R}$, update the state of $a$ and
$b$ to $δ(b,a,>)$. Conditioned on choosing a transition that includes
exactly one of $L$ or $R$,
each of the $n(n-1)$ ordered pairs of agents in the simulated
population protocol are equally likely.
So aside from an expected $O(n)$-factor slowdown from ineffective
transitions, the resulting $\BETWEEN$ protocol carries out precisely
the same transitions as the original $\POP^<$ protocol.

A similar approach can be used to simulate $\POP^=$ by separating the
equivalence classes by barrier tokens $|$, using transitions with
inputs $\Tuple{a,|,b}$ to detect when $a$ and $b$ are not in the same
class, and testing probabilistically if $a$ and $b$ are in the same
class by waiting long enough without seeing such a transition. But the
details are tedious, and the containment $\POP^< ⊆ \BETWEEN$ is
already enough to get:
\begin{theorem}
    \label{theorem-BPP-BETWEEN}
    $\BETWEEN=\classBPP$.
\end{theorem}
As noted in §\ref{section-simulating-SWSTS-in-BPP}, the $\BETWEEN⊆\classBPP$
inclusion requires a specific output convention for the
$\BETWEEN$ protocol, where a single agent reports
the output in its state. The theorem assumes that this convention applies.

\section{Conclusion}
\label{section-conclusion}

We have shown how extending the definition of a well-structured
transition system to include weights on configurations
and transition probabilities that are polynomial in these weights
allows proving general results about a large class of distributed
computing models with finite-state agents and random scheduling.
In particular, we can show that computations in such
stochastic well-structured transitions systems
necessarily terminate in polynomial time if they terminate at all,
which gives an exact characterization of the computational power of
many of these models in terms of standard complexity classes.

A natural question is what other properties of population protocols,
chemical reaction networks, and gossip models can be generalized to a
sufficiently abstract unified model? For example, the work of
Angluin~\etal~\cite{AngluinAER2007} characterizing the computational
power of stable computation in population protocols relies heavily on
the well-quasi-ordering of configurations of such protocols, but
requires additional additive structure that is captured by the more
general model of vector addition systems~\cite{KarpM1969}. Perhaps
similar results could be obtained for still more general models of
WSTSs. Similarly, the results of Mathur and
Ostrovsky~\cite{MathurO2022} on self-stabilizing population protocols
are ultimately based on the well-quasi-ordering of configurations, and
it seems likely that similar results could be obtained for WSTSs in
general.

For systems with randomized scheduling, our results give only upper
bounds on time. Is it possible to apply some sort of probabilistic scheduling
constraint to a well-structured transition system and get useful,
general lower bounds on the time to carry out particular computational
tasks, along the lines of known lower bounds for population protocols
and chemical reaction networks~\cite{DotyS2018,ChenCDS2017}? One
possibility might be to find an abstract notion of density (as
considered by Doty~\cite{Doty2014}) that is preserved by short enough
sequences of transitions, which could perhaps avoid some of the
complexity added to WSTSs by assigning explicit weights to
configurations and probabilities to transitions..

\bibliographystyle{alpha}
\bibliography{paper.bib}

\end{document}